\documentclass[11pt]{article}
\usepackage[T1]{fontenc}
\usepackage[utf8]{inputenc}
\usepackage{lmodern}
\usepackage{microtype}
\usepackage[margin=1in]{geometry}
\usepackage{amsmath,amssymb,amsthm}
\usepackage{graphicx,xcolor}
\usepackage{booktabs,float}
\usepackage{verbatim}
\usepackage{algorithm}
\usepackage{algorithmic}
\usepackage{xurl}
\usepackage[hidelinks]{hyperref}
\hypersetup{
  pdftitle={Denoising data using convex relaxations},
  pdfauthor={Charles Fefferman, Aalok Gangopadhyay, Matti Lassas, Jonathan Marty, Hariharan Narayanan}
}

\numberwithin{equation}{section}

\newcommand{\norm}[1]{\left\lVert#1\right\rVert}
\newcommand{\paran}[1]{\left(#1\right)}
\newcommand{\abs}[1]{\left |#1\right |}

\newcommand{\MM}{\mathcal{M}}
\newcommand{\NN}{\mathcal{N}}

\newcommand{\PP}{\mathcal{P}}
\newcommand{\R}{\mathbb{R}}
\newcommand{\E}{\mathbb{E}}
\newcommand{\F}{\mathcal{F}}
\newcommand{\op}{\operatorname{op}}
\newcommand{\conv}{\operatorname{conv}}
\newcommand{\Var}{\operatorname{Var}}
\newcommand{\diam}{\operatorname{diam}}
\newcommand{\reach}{\operatorname{reach}}
\newcommand{\mvolume}{\operatorname{volume}}
\newcommand{\Tan}{\operatorname{Tan}}
\newcommand{\lab}{\label}     \def\de{{\mathbf{\delta}}}   
  \def\ben{\begin{enumerate}}
\def\een{\end{enumerate}}
 \def\bit{\begin{itemize}}
\def\eit{\end{itemize}}
 \def\bel{\begin{lemma}} \def\eel{\end{lemma}}
\newcommand{\N}{\mathbb{N}} \newcommand{\Z}{\mathbb{Z}}   
    
   \newcommand{\p}{\mathbb{P}}
 \newcommand{\BB}{\mathcal B}   
  \newcommand{\la}{\lambda}  
   \def\eps{{\epsilon}}  \def\ie{i.\,e.\,}

\newcommand{\RR}{\mathbb{R}}

\DeclareMathOperator{\dist}{dist}

\renewcommand{\H}{\mathbb{H}}

\newcommand{\beqn}{\begin{equation}}
\newcommand{\eeqn}{\end{equation}}

\newcommand{\dd}{\mathbf{d}}
\newcommand{\supp}{\mathrm{supp}}

\theoremstyle{plain}
\newtheorem{theorem}{Theorem}[section]
\newtheorem{proposition}[theorem]{Proposition}
\newtheorem{corollary}[theorem]{Corollary}
\newtheorem{lemma}[theorem]{Lemma}

\theoremstyle{definition}
\newtheorem{definition}[theorem]{Definition}
\newtheorem{remark}[theorem]{Remark}
\newtheorem{observation}[theorem]{Observation}

\DeclareMathOperator{\argmax}{argmax}

\begin{document}

\title{Denoising data using convex relaxations}
\author{
Charles Fefferman\thanks{Department of Mathematics, Princeton University, Princeton, NJ 08544, USA.}
\and Aalok Gangopadhyay\thanks{TCS Research, Kolkata, West Bengal 700160, India.}
\and Matti Lassas\thanks{Department of Mathematics and Statistics, University of Helsinki, Helsinki FI-00014, Finland.}
\and Jonathan Marty\thanks{Program in Applied and Computational Mathematics, Princeton University, Princeton, NJ 08544, USA.}
\and Hariharan Narayanan\thanks{School of Technology and Computer Science, Tata Institute of Fundamental Research, Mumbai 400005, India.}
}
\date{}

\maketitle

\begin{abstract}
We study the problem of denoising observations \(Y_i=X_i+Z_i\), where the latent variables \(X_i\) are sampled from a low-dimensional manifold in \(\mathbb{R}^n\) and the noise variables \(Z_i\) are isotropic Gaussian. We propose a convex-relaxation estimator that first reduces dimension by principal component analysis and then projects the observations onto the convex hull of the projected latent manifold. We construct a statistical oracle that estimates its support function from empirical Gaussian tail probabilities of the noisy sample. Under a lower-mass condition on the latent distribution, we prove a uniform finite-sample guarantee over a directional net and derive error bounds for the resulting denoiser. The analysis combines risk bounds for least-squares projection under convex constraints (due to Chatterjee) with entropy bounds for the manifold. We also establish covering-number and Lipschitz estimates for a Cryo-electron microscopy observation model that imply a lower-mass condition for the image distribution under uniform orientations.
\end{abstract}



\section{Introduction}

In many modern applications of machine learning, computer vision, and scientific imaging, high-dimensional data is intrinsically low-dimensional, lying on or near some underlying geometric structure. This principle, commonly referred to as the manifold hypothesis, is central to numerous data analysis tasks. However, in practical scenarios, we rarely observe pristine samples from the manifold; rather, the data is typically corrupted by measurement noise. The problem of \emph{manifold denoising} aims to recover the underlying clean signals from these noisy observations.

Formally, suppose the clean data $X_1, \dots, X_N$ are sampled independently from a probability measure $\mu_0$ supported on a low-dimensional manifold $\mathcal{M}_0 \subset \mathbb{R}^n$. We observe noisy copies $Y_i = X_i + Z_i$, where $Z_i \sim \mathcal{N}(0, \sigma^2 I_n)$ is isotropic Gaussian noise. 

Our approach replaces the generally nonconvex task of projecting onto
$\MM_0$ by projection onto a convex relaxation. We first use PCA to identify
a lower-dimensional subspace and let $\Pi$ denote the corresponding
orthogonal projection. Within that subspace, our algorithm approximates the
Euclidean projection of each noisy observation onto
$K=\conv(\Pi\mathcal M_0)$ using an empirical support-function oracle.
Because $K$ is convex, its Euclidean distance function admits a dual
representation in terms of $h_K(v)=\sup_{x\in K}\langle x,v\rangle$. We
evaluate this representation over a finite directional net in the reduced
dimension.

Our algorithmic framework operates in two main stages:
\begin{enumerate}
    \item \textbf{Dimensionality Reduction:} We perform linear Principal Component Analysis (PCA) on a subset of the noisy samples to identify a low-dimensional subspace through the origin that optimally fits the data, rigorously bounding the empirical approximation error. This step eliminates noise in the orthogonal complement of the subspace and preserves the unit-ball normalization of the latent set.
    \item \textbf{Convex Projection via a Support-Function Oracle:} Within the
    reduced subspace, we estimate $h_K(v)$ from Gaussian tail masses beyond
    translates of a hyperplane normal to $v$. The signed quantity
    $\langle y,v\rangle-h_K(v)$, rather than an unsigned hyperplane distance,
    yields the dual formula for $\dist(y,K)$ and also covers observations that
    lie inside $K$.
\end{enumerate}


The primary contribution of this paper is the derivation of finite-sample
guarantees for this denoising procedure. We bound the sample complexity
required for the support-function oracle to achieve a specified accuracy.
Then, we apply Chatterjee's risk bounds for least squares under convex
constraints \cite{9a533cee-7bd2-34ce-a747-8254219de6bc}, combined with
Dudley's entropy integral and bounds on the metric entropy of the projected
latent support, to bound the distance between our algorithmic projection and
the latent sample. A linear functional has the same supremum over a set and
its convex hull, so the unrestricted Gaussian width needed for the risk
bound can be estimated directly on the latent support.

Finally, we instantiate our theoretical framework in the context of Cryo-Electron Microscopy (Cryo-EM), a domain characterized by extreme noise and Lie group symmetries. In Cryo-EM, 2D projection images of a 3D molecule are captured at unknown, random orientations. We model this physical process using the smooth action of the Lie group $G = SO(k)$, followed by a continuous X-ray transform and a discrete pixel sampling operator. We prove that this combined forward operator is Lipschitz continuous. After normalization, the pushforward of uniform orientations satisfies a lower-mass condition on the resulting image set. Our statistical guarantees therefore apply when the remaining noise, sampling, and parameter assumptions of the main theorem hold.

To place this contribution in context, it is useful to distinguish our goal---denoising individual observations---from the stronger task of fitting the entire latent manifold. Manifold fitting produces a global geometric estimate, for example in Hausdorff distance, whereas denoising asks only for accurate estimates of the sampled latent points. In statistics, Genovese, Perone-Pacifico, Verdinelli, and Wasserman formulated manifold estimation under Hausdorff loss and obtained minimax rates, thereby placing manifold fitting on a rigorous decision-theoretic footing \cite{GenovesePeronePacificoVerdinelliWasserman2012a,GenovesePeronePacificoVerdinelliWasserman2012b}. Aamari and Levrard later established nonasymptotic rates for estimating manifolds, tangent spaces, and curvature from samples \cite{AamariLevrard2019}. On the algorithmic side, Fefferman, Ivanov, Lassas, and Narayanan gave constructive procedures for fitting embedded manifolds from noisy data, first in a small-noise/large-reach regime \cite{FIMN} and later in the presence of arbitrarily large constant Gaussian noise \cite{large_noise}. The latter reconstruction result assumes that the manifold is $C^{2,1}$ and $R$-exposed, meaning that at every point of the manifold there is a tangent closed ball of radius $R$ that contains the manifold. Our denoising guarantee requires neither assumption. We do, however, extend a central technique from \cite{large_noise}: estimating support functions from empirical Gaussian tail probabilities. Here this idea is adapted to the weaker lower-mass setting and used for convex projection rather than full manifold reconstruction.

Other recent work continues to address manifold fitting. Yao, Su, Li, and Yau proposed a two-step method with theoretical guarantees \cite{YaoSuLiYau2023}, while Yao, Su, and Yau introduced a neural generative approach based on CycleGAN that learns smooth maps between latent and ambient spaces and supports projection onto the fitted manifold \cite{YaoSuYau2024}; see also Yao and Xia for a related unbounded-noise procedure based on tangent-space estimation at projected points \cite{YaoXia2025}.

The distinction between denoising and full manifold reconstruction motivates the weaker regularity assumption used below. Rather than impose global smoothness or exposedness, we require only that small balls centered on the latent support carry uniformly positive probability mass. Specifically, let $\mu_0$ be a probability measure with compact support $\MM_0:=\supp(\mu_0)\subseteq B_n(0,1)$, where $B_n(0,1)$ is the origin-centered closed Euclidean ball of radius $1$ in $\R^n$. We assume that a valid exponent $d\in\Z_+$ and a positive lower-mass constant $c_{\MM_0}$ are supplied, so that for every $\eps\in(0,2]$ and every $x\in\MM_0$,
\begin{align}
\lab{eq:basic}
\mu_0(B(x,\eps)) > c_{\MM_0}\omega_d\eps^d,
\end{align}
where $\omega_d$ is the Lebesgue measure of the unit ball in $d$-dimensional Euclidean space. We may assume $c_{\MM_0}^{-1}>e$.
We remark that if $\MM_0$ is a compact $C^2$-smooth $d$-dimensional manifold without boundary, of volume $V$ and reach at least $\tau$, and $\mu_0$ is the uniform measure on $\MM_0$, then $c_{\MM_0}$ can be chosen to be at least $\frac{c^d \tau^d}{V}$ for some absolute constant $c$.
The lower-mass condition holds in significantly greater generality. For example, the pushforward of a measure $\nu_0$ under a Lipschitz map retains this property, with an adjusted constant, whenever $\nu_0$ satisfies an analogous bound. This allows us to treat measures on image sets that need not be smooth manifolds, including the Lipschitz images of $SO(3)$ encountered in the Cryo-EM application of Section~\ref{sec:Cryo-EM}.
    
\subsection{Preliminaries}
    We use $c, C, C_1,$ etc. to denote absolute constants. For a function $H$, we use the notation $H = \tilde{\Theta}(H')$ if $\log H = \Theta(\log H')$ \ie, if there exist positive universal constants $c, C$ such that the following inequality holds for all arguments in the domain of $H$: $c \log H' \leq \log H \leq C \log H'$.

    Let $N$ be some sufficiently large positive integer.  Let $(X_1, \dots, X_N)$ be a sequence of i.i.d.\ samples from $\mu_0$. Let $(Z_1, \dots, Z_N)$ be a sequence of i.i.d.\ samples from $\NN(0, \sigma^2 I_n)$, that is independent of $(X_1, \dots, X_N)$. 
    For $i \in \N$, let $Y_i  = X_i + Z_i.$ We consider a scenario where we do not have direct access to $(X_1, \dots, X_N)$ and $(Z_1, \dots, Z_N)$, but wish to infer approximations of the $X_i$ for $i > N_1$ from $(Y_1, \dots, Y_{N_1})$ and $Y_i$ alone. \textcolor{black}{Here $0 < N_0 < N_1 < N$ are integers fixing a partition of the observations into three disjoint blocks: the first $N_0$ samples are used for PCA (Algorithm~1), the next $N_1-N_0$ samples serve as the oracle batch used by the support-function oracle, and the remaining $N-N_1$ samples are the ones that are denoised.}

\medskip
\noindent\textbf{Preview of the Main Result.}

To provide immediate context for the algorithm, we state an informal version of our main theoretical guarantee here. The complete statement is deferred to Section~\ref{sec:main}.

At a high level, our analysis bounds the total denoising error by decoupling it into three distinct components: the bias from the initial PCA dimensionality reduction, the fundamental statistical risk of projecting onto the convex hull, and the algorithmic approximation error introduced by our finite-sample support-function oracle.

\newpage
\noindent \textbf{Main Theorem} (informal):\\
Let $\mu_0$ have compact support $\MM_0=\supp(\mu_0)\subseteq B_n(0,1)$, and suppose known $d\in\Z_+$ and $c_{\MM_0}>0$ satisfy, for all $\eps\in(0,2]$,
\begin{align*}  \mu_0(B(x,\eps))>c_{\MM_0}\omega_d\eps^d \end{align*}
for every $x\in\MM_0$, where $\omega_d$ is the volume of the unit ball
in $\R^d$. We may assume $c_{\MM_0}^{-1}>e$.

    
    Let $N$ be some sufficiently large positive integer.  Let $(X_1, \dots, X_N)$ be a sequence of i.i.d.\ samples from $\mu_0$. Let $(Z_1, \dots, Z_N)$ be a sequence of i.i.d.\ samples from $\NN(0, \sigma^2 I_n)$, that is independent of $(X_1, \dots, X_N)$. 
    For $i \in \N$, let $Y_i  = X_i + Z_i.$

Assume $0<\sigma\leq1$, and fix a target index $i>N_1$, writing
$X=X_i$ and $Y=Y_i$. This target pair is independent of the PCA and
oracle batches. For a user-prescribed algorithmic tolerance $\epsilon>0$,
suppose Algorithm~1 uses $N_0$ PCA samples and $N_{\mathrm{oracle}}$
oracle samples, and all parameter and sample-size conditions of
Theorem~\ref{thm:main} hold. In particular, the oracle tolerance $\delta$ satisfies
\eqref{eq:oracle-parameter-regime}, and $\delta,\rho$ satisfy
\eqref{eq:main-projection-parameters}, which links them to $\epsilon$
and the target-norm failure probability $\beta$.

With probability at least
$1-C\alpha-3\exp\left(-\frac{\gamma^4}
{32(1+\gamma/(2J)^{1/4})^2}\right)-\eta-\beta$, the estimated denoised
point $\hat{X}_{algo}$ produced by Algorithm 1 satisfies:
\begin{equation}
    \|\hat{X}_{algo} - X\|_2 \le \underbrace{2\left(\frac{4\epsilon_0^2 + 2\epsilon_{emp}}{c_{\MM_0}\omega_d}\right)^{\frac{1}{d+2}}}_{\text{PCA Subspace Bias}} + \underbrace{\sigma\left((2J)^{1/2}+\gamma(2J)^{1/4}\right)}_{\text{Statistical Risk on } K} + \underbrace{\epsilon}_{\text{Algorithmic Error}},
\end{equation}

\begin{table}[H]
\centering
\caption{Notation used throughout the paper.}
\label{tab:notation}
\small
\begin{tabular}{p{0.22\textwidth} p{0.70\textwidth}}
\toprule
\textbf{Symbol} & \textbf{Meaning} \\
\midrule

$\MM_0 \subset \mathbb{R}^n$ & Compact latent support inside $B_n(0,1)$ \\
$\mu_0$ & Probability measure with support $\MM_0$ \\
$d$ & Known positive integer exponent in the lower-mass condition \\
$\omega_d$ & Volume of unit ball in $\mathbb{R}^d$ \\
$c_{\MM_0}$ & Known positive lower-mass constant \\

$X_i$ & Latent samples, $X_i \sim \mu_0$ \\
$Z_i$ & Gaussian noise, $Z_i \sim N(0,\sigma^2 I_n)$ \\
$Y_i$ & Observations, $Y_i = X_i + Z_i$ \\
$\sigma$ & Known noise level (standard deviation) \\
$\delta$ & Localization radius and histogram-bin width in Algorithm~3 \\
$\rho$ & Radius of the directional net on $S^{D-1}$ \\

$N,N_0,N_1$ & Total sample size and boundaries of the PCA and oracle batches \\
$N_{\mathrm{oracle}}$ & Oracle sample size $N_1-N_0$ \\
$\mathfrak{S}$ & PCA subspace of dimension $D$ \\
$\Pi$ & Orthogonal projection onto $\mathfrak{S}$ \\
$D$ & Reduced dimension chosen via volumetric scaling \\

$\MM$ & Projected latent support $\Pi(\MM_0)$ \\
$K=\operatorname{conv}(\MM)$ & Convex hull of $\MM$ \\
$P_K(y)$ & Euclidean projection of $y$ onto $K$ \\

$H$ & Half-space in Algorithms~2--3; hyperplane in Section~\ref{sec:HypDist} \\
$\operatorname{dist}(H,\MM)$ & Euclidean distance between $H$ and $\MM$ \\
$\Gamma(H)$ & Gaussian integral over hyperplane $H$ \\
$h_K(v)$ & Support function $\sup_{x\in K}\langle x,v\rangle$ \\
$\widehat h(v)$ & Empirical estimate of the directional support value \\
$g_y(v)$ & Signed separation $\langle y,v\rangle-h_K(v)$ \\

$K_\sigma,\MM_\sigma$ & Rescaled sets $\sigma^{-1}K,\sigma^{-1}\MM$ \\
$x_\sigma,G$ & Rescaled clean point $\sigma^{-1}x$ and standard Gaussian vector \\
$X_u$ & Gaussian process $\langle G,u-x_\sigma\rangle$, $u\in\MM_\sigma$ \\
$d(u,v)$ & Canonical process metric $\|u-v\|_2$ \\
$N_C(r,A)$ & Minimal cardinality of an $r$-net centered in $A$ \\
$N_P(r,A)$ & Maximal cardinality of an $r$-separated subset of $A$ \\

$F$ & X-ray transform \\
$S$ (sampling) & Pixel sampling operator \\
$T$ & Group action (e.g.\ $SO(k)$) \\
$\mathrm{Lip}(\cdot)$ & Lipschitz constant \\

\bottomrule
\end{tabular}
\end{table}
 
     \section{Algorithm}\label{sec:Algo}
     
       \subsection{Algorithm $\mathrm{Denoise}(Y_1, \dots, Y_N)$:}
 Let $\eps_0 > 0$ be an error parameter. 

   
    \begin{algorithm}[H]
\caption{$\mathrm{Denoise}(Y_1,\dots,Y_N)$}
\begin{algorithmic}[1]
\STATE Perform PCA on $(Y_1,\dots,Y_{N_0})$ to find an optimal least-squares fit subspace $\mathfrak{S}$ of dimension $D := \min(n, \lceil c_{\MM_0}^{-1} \omega_d^{-1} \eps_0^{-d} \rceil )$.  Let $\MM := \Pi_{\mathfrak{S}} \MM_0$.  Let $K := \text{conv}(\MM)$. 
\FOR{\textcolor{black}{$i = N_0 + 1$ to $N_1$}}
    \STATE \textcolor{black}{Set $\widetilde{Y}_i \gets \Pi_{\mathfrak{S}}Y_i$ \quad // oracle batch}
\ENDFOR
\FOR{$i = N_1 + 1$ to $N$}
    \STATE Set $\widetilde{Y}_i \gets \Pi_{\mathfrak{S}}Y_i$
\ENDFOR
\FOR{$i = N_1 + 1$ to $N$}
    \STATE Set $\hat{X}_i \gets \mathrm{Proj}_{K}(\widetilde{Y}_i, \{\widetilde{Y}_{N_0+1}, \dots, \widetilde{Y}_{N_1}\})$
\ENDFOR
\RETURN $(\hat{X}_{N_1 + 1}, \dots, \hat{X}_N)$
\end{algorithmic}
\end{algorithm}

The PCA error and its dependence on $N_0$ are given in
Proposition~\ref{prop:1}. The oracle batch size
$N_{\mathrm{oracle}}=N_1-N_0$ is chosen according to
Theorem~\ref{thm:uniform-support-oracle}. The remaining $N-N_1$ samples
are denoised.

    \subsection{Algorithm $\mathrm{Proj}_{K}(\widetilde{Y}, P)$:}
    
        
        Let $H_v := \{x|\langle x, v\rangle \leq 0\}$\textcolor{black}{, a closed half-space of $\R^D$ with outward normal $v$}. Given a random variable \textcolor{black}{$\widetilde Y$} taking values in $\R^D$,  we also use the notation,  \textcolor{black}{$\widetilde Y + H_v := \{x + \widetilde Y \mid \langle x, v\rangle \leq 0\}$} (i.e.\ the half-space translated so that its boundary passes through \textcolor{black}{$\widetilde Y$}).
        
The maximization below uses a fixed finite $\rho$-net
$\mathcal V_1\subset S^{D-1}$. The choices of $\rho$ and the oracle
tolerance are specified in Theorem~\ref{thm:main}.

\begin{algorithm}[H]
\caption{$\mathrm{Proj}_{K}(\widetilde{Y}, P)$}
\begin{algorithmic}[1]
\STATE Set $v_0 \gets \argmax\limits_{v \in \mathcal V_1} \mathrm{Dist}_{K}(\widetilde{Y} + H_v, P)$
\STATE Set $\la_0 \gets \max\{0,\mathrm{Dist}_{K}(\widetilde{Y} + H_{v_0}, P)\}$
\RETURN $\widetilde{Y} - \la_0 v_0$
\end{algorithmic}
\end{algorithm}

    \subsection{Algorithm $\mathrm{Dist}_{K}(H, P)$:}

For $H=\{x:\langle x,b\rangle\le t\}$, the algorithm below
estimates the signed separation $t-h_K(b)$. Its uniform $3\delta$-accuracy
guarantee, with failure probability at most $\eta$, is given in
Theorem~\ref{thm:uniform-support-oracle}.

We introduce two thresholds used by the \textcolor{black}{Algorithm below} for $\mathrm{Dist}_{K}(H, P)$.

Let
\[
\kappa_0 := \sqrt{2\pi}\sigma,
\qquad
\kappa_1 := \frac{ c_{\MM_0}^{-1}}{\delta^d \omega_d } \sqrt{2\pi}\sigma.
\]

Define
\[
r_\delta := \frac{\sigma^2}{\delta} \log\!\left(\frac{\kappa_1}{\kappa_0}\right) = \frac{ \sigma^2}{\delta}\log\!\left(\frac{ c_{\MM_0}^{-1}}{\delta^d \omega_d } \right).
\]

We then set
\[
\Gamma_\delta := \kappa_1^{-1}
\exp\!\left(-\frac{r_\delta^2}{2\sigma^2}\right).
\]

The algorithm estimates the Gaussian-convolved density along the normal
direction $b$ by discretizing the tail distribution of the samples. Define
\[
F_{\gamma}
:= 
\frac{\left|\{\, X \in P : \langle X, b \rangle > \gamma \,\}\right|}
{|P|}.
\]

This quantity estimates the probability mass of the noisy distribution lying
beyond the hyperplane
\[
\langle x, b \rangle = \gamma .
\]

\begin{algorithm}[H]
\caption{$\mathrm{Dist}_{K}(H, P)$}
\begin{algorithmic}[1]
\STATE Set $b \in S^{D-1}$ and $t \in \R$ such that $H = \{x \in \R^D: \langle x, b \rangle \leq t\}$
\STATE Denote $F_\gamma = \frac{1}{|P|} \#\{X \in P : \langle X,b\rangle > \gamma\}$
\STATE Denote $\Gamma^{est}_j = \frac{1}{\delta}\left(F_{j\delta} - F_{(j+1)\delta}\right)$
\STATE Set $j \gets \textcolor{black}{\lceil \delta^{-1}\rceil + 1}$ \textcolor{black}{// start beyond the unit ball}
\WHILE{$\Gamma^{est}_j > \Gamma_\delta$}
\STATE Set $j \gets j + 1$
\ENDWHILE
\RETURN $t-(j\delta - r_\delta)$
\end{algorithmic}
\end{algorithm}

\begin{figure}
    \centering
    \includegraphics[width=4in]{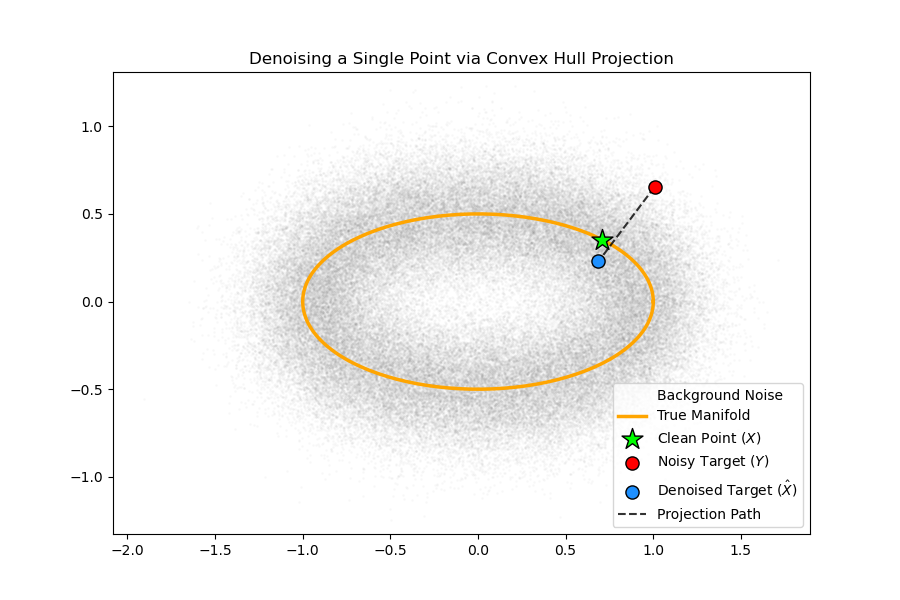} 
    \caption{A stylised visual for the experiment with $\mu_0$ the law of $(\cos t,\tfrac12\sin t)$ for $t$ uniform on $[0,2\pi)$. Its support is the ellipse $\{(x,y):x^2+4y^2=1\}$, and the Gaussian noise has standard deviation $\sigma=0.2$.
} \label{fig:Ellipse}
\end{figure}   

Appendix~\ref{sec:ellipsoid-experiment} reports an end-to-end ellipsoid experiment.

\subsection{Extent of noise reduction}
Condition on the PCA batch and let $\Pi$ be the resulting projection.
For a fresh target, write $x=\Pi X\in K$, $y=\Pi Y=x+\Pi Z$, and
$\widehat x=P_K(y)$. The variational inequality for projection
onto a closed convex set gives
\[
 \langle y-\widehat x,x-\widehat x\rangle\leq0,
 \qquad
 \|\widehat x-x\|_2\leq\|\Pi Z\|_2.
\]
Thus exact convex projection cannot increase the error relative to the
projected noisy observation.

For $0<\sigma\leq1$, the Gaussian-width bound in
Lemma~\ref{lemma:RDEI_convex_hull} gives
\begin{equation*}
   J=\frac{C}{\sigma}
   \sqrt{\log\!\left(\frac{1}{c_{\MM_0}\omega_d}\right)}.
\end{equation*}
The unrestricted width satisfies $W_\sigma(K)=W_\sigma(\MM)$ because
linear functionals have the same supremum on $\MM$ and its convex hull.
The width bound also gives a direct expectation bound. In the dimensionless coordinates
of Section~\ref{sec:6}, optimality of exact projection and
$x_\sigma\in K_\sigma$ imply
\[
 \|\widehat x_\sigma-x_\sigma\|_2^2
 \leq2\langle G,\widehat x_\sigma-x_\sigma\rangle.
\]
Taking expectations and applying Lemma~\ref{lemma:RDEI_convex_hull} yields
\[
 \E\|\widehat x-x\|_2^2
 \leq2\sigma^2 W_\sigma(K)\leq2\sigma^2J,
 \qquad
 \E\|\widehat x-x\|_2\leq\sigma\sqrt{2J}.
\]
These estimates are uniform in the clean point and the PCA subspace, so
we may average over the fresh clean target and the PCA batch. Since
$\E\|Y-X\|_2\geq\sigma\sqrt{2n/\pi}$, the statistical error of exact
projection satisfies
\[
 \frac{\E\|P_K(\Pi Y)-\Pi X\|_2}
      {\E\|Y-X\|_2}
 \leq\sqrt{\frac{\pi J}{n}}.
\]
This bound gives a reduction relative to the original ambient noise when
$J$ is sufficiently small compared with $n$. Its dependence on
$\log(1/(c_{\MM_0}\omega_d))$ is essential; a rate that omits this term
requires an additional quantitative hypothesis on the lower-mass constant.
Theorem~\ref{thm:main} includes the PCA bias and the finite-sample
algorithmic error in the full denoising guarantee.

\subsection{The projection $\Pi$}\lab{ssec:Pi}

Let $B(x, R)$ denote the Euclidean ball of radius $R$ centered at a point $x$, in some ambient Euclidean space that will be clear from context.
For a linear subspace $S$ and its orthogonal projection $\Pi_S$, set
\[
 e(S):=\sup_{x\in\MM_0}\|x-\Pi_Sx\|_2
      =\sup_{x\in\MM_0}\dist(x,S).
\]
Then
\[
 d_{\mathrm{H}}(\MM_0,\Pi_S\MM_0)\le e(S),
 \qquad
 d_{\mathrm{H}}(\conv(\MM_0),\conv(\Pi_S\MM_0))\le e(S).
\]
We will prove a probabilistic upper bound on $e(S)$ for the linear PCA
subspace.

Suppose 
$
R = C\sigma\sqrt n + C\sigma\sqrt{\log(CN_0/\alpha)}.
$
We have (as in the proof of Claim 3.3 in \cite{FIMN})  that
\begin{equation}\label{eq:IRbound_goal}
\overline{I}_R:=\int_{\{|x|>R\}} (2\pi\sigma^2)^{-n/2}\exp\!\left(-\frac{|x|^2}{2\sigma^2}\right)\,dx
\;<\; 1-\bigl(1-\alpha/2\bigr)^{1/N_0}.
\end{equation}
Indeed, the left-hand side $\overline{I}_R$ can be bounded above as follows:
\begin{align}
\overline{I}_R \, e^{R^2/(4\sigma^2)}
&\le \int_{\mathbb R^n} (2\pi\sigma^2)^{-n/2}
\exp\!\left(-\frac{|x|^2}{2\sigma^2}\right)\exp\!\left(\frac{|x|^2}{4\sigma^2}\right)\,dx \notag\\
&= 2^{n/2}. \label{eq:IR_exp_trick}
\end{align}

and so, using \eqref{eq:IR_exp_trick},
\begin{align}
\overline{I}_R
&\le 2^{n/2}\exp\!\left(-\frac{R^2}{4\sigma^2}\right)
\le 2^{n/2}\exp\!\left(-Cn - C\log(CN_0/\alpha)\right) \notag\\
&\le \frac{\alpha}{C N_0}
\le 1-\bigl(1-\alpha/2\bigr)^{1/N_0}. \notag
\end{align}
\textcolor{black}{(The last step uses the elementary inequality $1-(1-x)^{1/N_0}\ge x/N_0$ for $x\in(0,1)$, applied with $x=\alpha/2$, provided $C$ is chosen a sufficiently large absolute constant.)}

Let $Z_1,\dots,Z_{N_0}$ be i.i.d.\ random vectors in $\mathbb{R}^n$ with
\[
Z_i \sim \NN(0,\sigma^2 I_n).
\]
Assume that
\[
\mathbb{P}\!\left( Z_1,\dots,Z_{N_0} \in B(0,R) \right)
\;\ge\;
1-\frac{\alpha}{2}.
\]

Let $\NN^R(0,\sigma^2 I_n)$ denote the truncated Gaussian measure normalized to be a probability measure,
that is, the distribution
\[
\NN^R(0,\sigma^2 I_n)
=
(1 - \overline{I}_R)^{-1}\NN(0,\sigma^2 I_n)\big|_{B(0,R)} .
\]

Then there exists $\alpha_0 \le \frac{\alpha}{2}$ and a probability measure
$\nu_0$ on $(\mathbb{R}^n)^{N_0}$ such that
\[
\big(\NN(0,\sigma^2 I_n)\big)^{\otimes N_0}
=
(1-\alpha_0)\,
\big(\NN^R(0,\sigma^2 I_n)\big)^{\otimes N_0}
\;+\;
\alpha_0\,\nu_0.
\]

Let $\mu$ and $\mu_R$ be the measures given by \[\mu = \mu_0 \ast \NN(0, \sigma^2 I_n)\] and \[\mu_R =  \mu_0 \ast \NN^R(0, \sigma^2 I_n).\] Since $\supp(\mu_0) \subseteq B(0, 1)$,   $\mu_R$ is supported on $B(0, R+1):=\{x \in \RR^n \big|\, \|x\| \leq R+1\}$.

Let $\PP$ be a probability distribution supported on $B(0, R+1).$  Let $Y_1, \dots, Y_{N_0}$ be i.i.d samples from $\PP$.
Let $\H := \H_D$ be the set whose elements are affine subspaces $S \subseteq \R^n$ of dimension $ D$, each of which  intersects $B:= B(0, 1)$.
Let $\F_{D}$ be the set of all loss functions $F(x) =  \dist(x, H)^2$ for some $H \in \H$
(where $\dist(x, S) := \inf_{y \in S} \|x - y\|$).
We wish to obtain a probabilistic upper bound on
\begin{align} \label{eq1}\sup_{F \in \F_{D}} \Bigg | \frac{\sum_{i=1}^{N_0}  F(\frac{Y_i}{R+1})}{N_0} - \E_{Y \sim \PP} F(\frac{Y}{R+1})\Bigg |. \end{align}
We quote a probabilistic upper bound on (\ref{eq1}) which appeared as Lemma 3.2 in \cite{FIMN}:

\begin{lemma}\label{lem:kplanes}
 Let $Y_1, \dots, Y_{N_0}$ be i.i.d samples from $\PP$, a distribution supported on the ball of radius $R+1$ in $\R^n$.
 
 Then, firstly,
\begin{align*} \p\left[\left\|\frac{\sum_{i=1}^{N_0} \frac{Y_i}{R+1}}{N_0} - \E_{Y \sim \PP} \frac{Y}{R+1}\right\| \leq 2 \left(\sqrt{\frac{1}{N_0}}\right)\left(1 + \sqrt{2 \ln (4/\alpha)}\right)\right] > 1 - \alpha. \end{align*}
Secondly, 
\[
\begin{aligned}
&\p\Bigg[\sup_{F \in \F_D}\Bigg|
 \frac{1}{N_0}\sum_{i=1}^{N_0}F\!\left(\frac{Y_i}{R+1}\right)
 - \E_{Y \sim \PP}F\!\left(\frac{Y}{R+1}\right)\Bigg| \\
&\hspace{4em}\leq
 2\frac{\sqrt D+2}{\sqrt{N_0}}
 \left(1+\sqrt{2\ln(4/\alpha)}\right)\Bigg]
 > 1-2\alpha.
\end{aligned}
\]
\end{lemma}

 In the algorithm, we perform linear PCA on $(Y_1,\dots,Y_{N_0})$ to find
an optimal least-squares fit linear subspace $S$ of dimension
$D := \min(n, \lceil c_{\MM_0}^{-1} \omega_d^{-1}\eps_0^{-d} \rceil )$.
 Let \[\epsilon_{emp} := 2(R+1)^2 \left({\frac{\sqrt{D}+2}{\sqrt{N_0}}}\right)\left(1 + \sqrt{2 \ln (4/\alpha)}\right).\]

\begin{proposition}[Uniform accuracy of the linear PCA subspace]\lab{prop:1}
Let $\widehat S$ be the $D$-dimensional linear subspace returned by the PCA
step, let $\widehat\Pi$ be its orthogonal projection, and define
\[
 B_{\mathrm{PCA}}:=4\epsilon_0^2+2\epsilon_{emp}.
\]
Then, with probability at least $1-C\alpha$,
\begin{equation}\label{eq:correct-pca-bound}
 \sup_{x\in\MM_0}\|x-\widehat\Pi x\|_2
 \le
 2\left(\frac{B_{\mathrm{PCA}}}{c_{\MM_0}\omega_d}\right)^{1/(d+2)}.
\end{equation}
Consequently, both
$d_{\mathrm{H}}(\MM_0,\widehat\Pi\MM_0)$ and
$d_{\mathrm{H}}(\conv(\MM_0),\conv(\widehat\Pi\MM_0))$ are bounded by the
right-hand side of~\eqref{eq:correct-pca-bound}. Moreover,
$\widehat\Pi\MM_0\subseteq B_D(0,1)$.
\end{proposition}

\begin{proof}
If $\epsilon_0>2$, the right-hand side of~\eqref{eq:correct-pca-bound}
exceeds $1$, because $c_{\MM_0}\omega_d<1$ and
$B_{\mathrm{PCA}}\geq4\epsilon_0^2$. The conclusion then follows
directly from $\sup_{x\in\MM_0}\dist(x,\widehat S)\leq1$.
We therefore assume $0<\epsilon_0\leq2$.
Let $\mathcal B$ be a maximal collection of pairwise disjoint open
$\epsilon_0$-balls centered on $\MM_0$. By the lower-mass condition,
\[
 |\mathcal B|\le
 (c_{\MM_0}\omega_d\epsilon_0^d)^{-1}.
\]
Maximality also implies that the concentric $2\epsilon_0$-balls cover
$\MM_0$. If $D<n$, the definition of $D$ gives
$|\mathcal B|\le D$, so the linear span of the centers can be extended to a
$D$-dimensional subspace $S_{\mathcal B}$. If $D=n$, take
$S_{\mathcal B}=\mathbb R^n$. In either case,
\begin{equation}\label{eq:packing-candidate-bias}
 \sup_{x\in\MM_0}\dist(x,S_{\mathcal B})\le2\epsilon_0.
\end{equation}

Let $W\sim\NN^R(0,\sigma^2I_n)$. Rotational symmetry gives
$\mathbb EW=0$ and $\mathbb EWW^\top=\lambda_RI_n$ for some
$0\le\lambda_R\le\sigma^2$. Therefore, for every $D$-dimensional linear
subspace $S$ and independent $X\sim\mu_0$,
\begin{equation}\label{eq:truncated-risk-decomposition}
 \mathbb E\dist(X+W,S)^2
 =\mathbb E\dist(X,S)^2+(n-D)\lambda_R.
\end{equation}
In particular,~\eqref{eq:packing-candidate-bias} gives
\begin{equation}\label{eq:candidate-pca-risk}
 \mathbb E\dist(X+W,S_{\mathcal B})^2
 \le(n-D)\lambda_R+4\epsilon_0^2.
\end{equation}

Apply Lemma~\ref{lem:kplanes} to $\mu_R$. Because linear subspaces are a
subclass of the affine family in that lemma and scaling commutes with a
linear subspace, with probability at least $1-2\alpha$ we have,
simultaneously for every $D$-dimensional linear $S$,
\begin{equation}\label{eq:uniform-linear-pca-risk}
 |R_{emp}(S)-R_R(S)|\le\epsilon_{emp},
\end{equation}
where $R_{emp}$ is empirical squared distance and
$R_R(S)=\mathbb E_{X,W}\dist(X+W,S)^2$. On this event, empirical optimality
of $\widehat S$,~\eqref{eq:candidate-pca-risk}, and
\eqref{eq:uniform-linear-pca-risk} imply
\begin{align*}
 R_R(\widehat S)
 &\le R_{emp}(\widehat S)+\epsilon_{emp}\\
 &\le R_{emp}(S_{\mathcal B})+\epsilon_{emp}\\
 &\le(n-D)\lambda_R+4\epsilon_0^2+2\epsilon_{emp}.
\end{align*}
Subtracting the common Gaussian term using
\eqref{eq:truncated-risk-decomposition} yields
\begin{equation}\label{eq:mean-pca-bias}
 \mathbb E_{X\sim\mu_0}\dist(X,\widehat S)^2\le B_{\mathrm{PCA}}.
\end{equation}
The mixture representation preceding Lemma~\ref{lem:kplanes} transfers this
event from the truncated experiment to the original Gaussian experiment at
an additional failure probability at most $\alpha/2$. Thus
\eqref{eq:mean-pca-bias} holds with probability at least $1-C\alpha$ for the
subspace returned by Algorithm~1.

It remains to convert the mean bound into a uniform one. Put
$s=\sup_{x\in\MM_0}\dist(x,\widehat S)$. Since $\widehat S$ is linear and
$\MM_0\subseteq B_n(0,1)$, we have $s\le1$. If $s=0$, the desired
bound is immediate. Otherwise, for $0<\xi<s$, choose
$x_\xi\in\MM_0$ with
$\dist(x_\xi,\widehat S)\ge s-\xi$. The distance function is
$1$-Lipschitz, so every
$x\in B(x_\xi,(s-\xi)/2)$ satisfies
$\dist(x,\widehat S)\ge(s-\xi)/2$. The lower-mass condition and
\eqref{eq:mean-pca-bias} therefore give
\[
 B_{\mathrm{PCA}}\ge
 c_{\MM_0}\omega_d\left(\frac{s-\xi}{2}\right)^{d+2}.
\]
Letting $\xi\downarrow0$ proves~\eqref{eq:correct-pca-bound}. The Hausdorff
bounds follow by projecting points and convex combinations. Finally, linear
orthogonal projection is contractive from the origin, hence
$\widehat\Pi\MM_0\subseteq B_D(0,1)$.
\end{proof}

\section{Hypocycloids}

In order to build intuition for denoising a Gaussian-corrupted distribution
via projection onto the convex hull of its support, we consider
hypocycloid-supported distributions.  A hypocycloid is traced by a point on a
circle of radius $b<a$ rolling inside a circle of radius~$a$.  It has the
parametric form

\begin{equation}\begin{split}
  x(t) &= (a-b)\cos\left(t\right) + b \cos\left(\frac{a-b}{b}t\right)\\
  y(t) &= (a-b)\sin\left(t\right) - b \sin\left(\frac{a-b}{b}t\right)
\end{split}\end{equation}

When $a/b=n$ is an integer with $n\ge3$, the curve has $n$ cusps, and its convex hull is a
regular $n$-gon.  Projection onto this polygon maps a point in the exterior
normal cone of an edge to that edge and a point in the normal cone of a vertex
to the vertex; points already in the polygon are unchanged.

We take $a=1$ and $a/b=5$, so the convex hull is a regular pentagon, and apply
Gaussian noise to three source locations: a cusp, a point near a cusp, and the
middle of an arc.  Figures~\ref{fig:HypocycloidCusp}--
\ref{fig:HypocycloidMidArc} report $10{,}000$ draws for each source location
and each of $\sigma\in\{0.3,0.5,1,10\}$.  In every row, the left panel shows
the noisy sample, with points inside the pentagon in blue and points outside
it in purple; the orange curve is the hypocycloid and the red point is the
source.  The center panel shows where the exterior points project onto the
boundary, with color indicating multiplicity, and the right panel gives the
corresponding boundary-location histogram.  Both geometric panels use equal
horizontal and vertical coordinate scales; hence the fivefold geometry is not
visually stretched.  
\begin{figure}[p]
    \centering
    \includegraphics[width=0.92\linewidth]{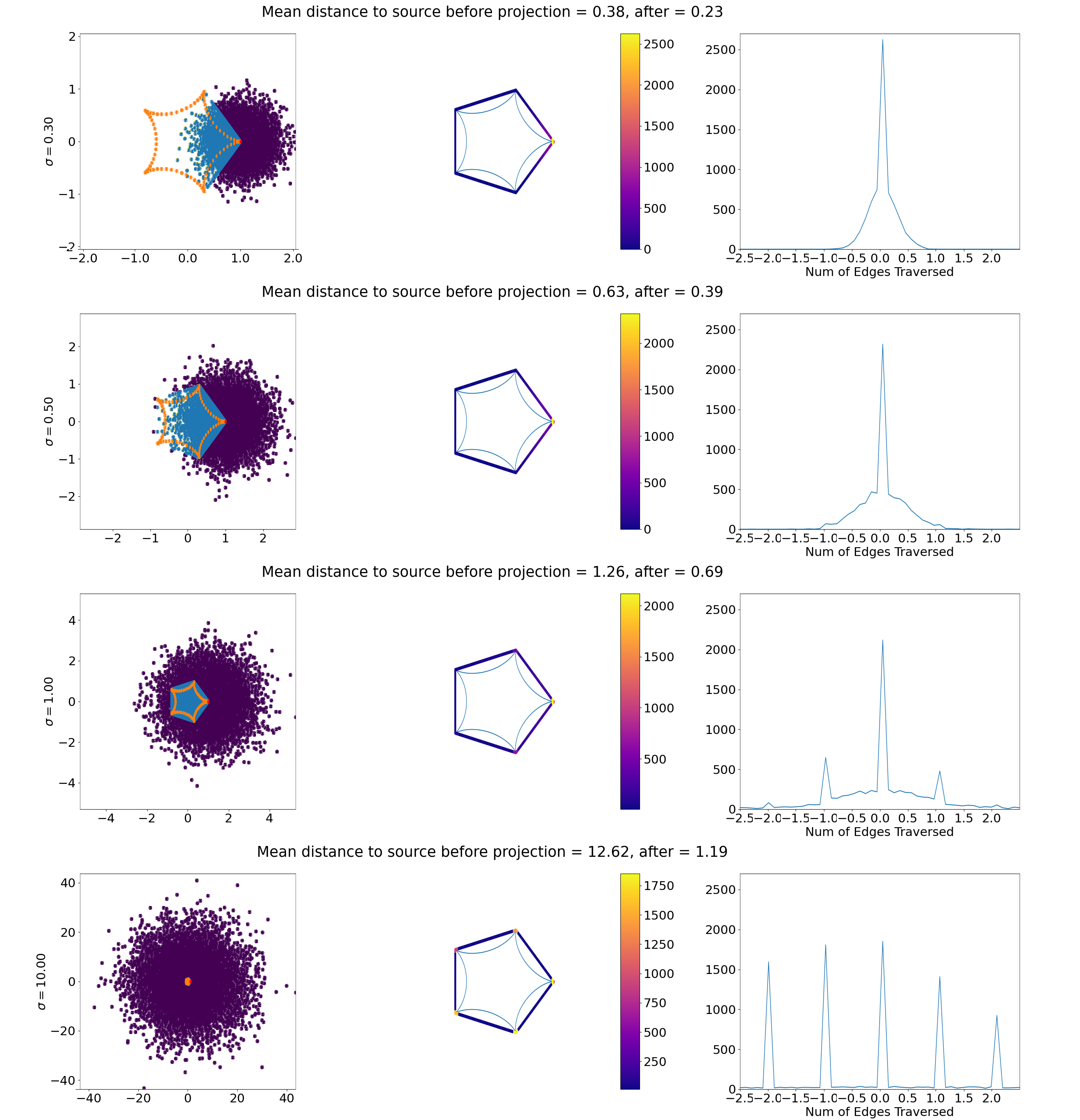}
    \caption{Hypocycloid experiment with the source point at a cusp.  Rows
    correspond to $\sigma=0.3,0.5,1,$ and $10$, from top to bottom, with
    $10{,}000$ draws per row. }
    \label{fig:HypocycloidCusp}
\end{figure}

\begin{figure}[p]
    \centering
    \includegraphics[width=0.92\linewidth]{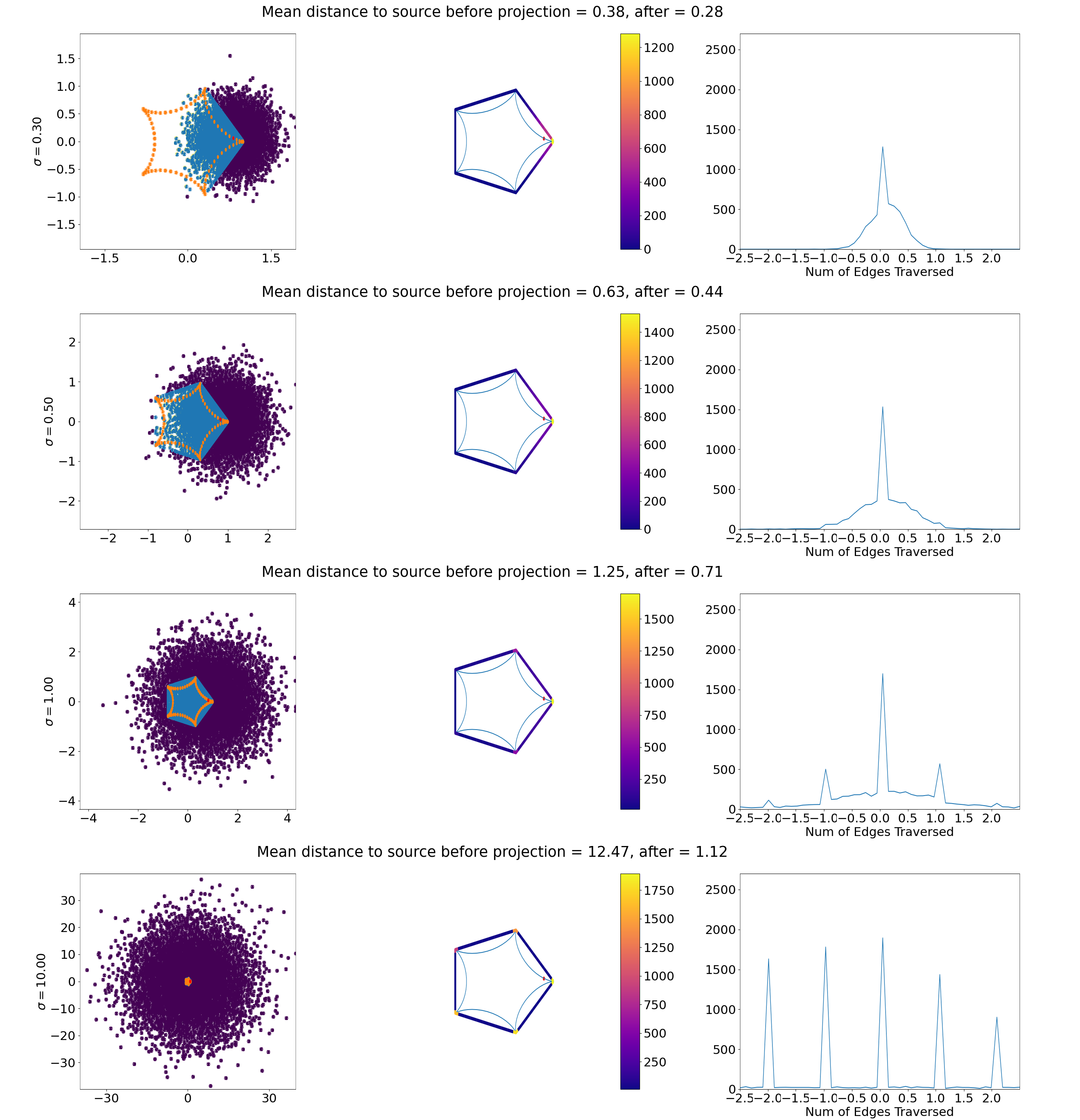}
    \caption{Hypocycloid experiment with the source point near a cusp.  Rows
    correspond to $\sigma=0.3,0.5,1,$ and $10$, from top to bottom, with
    $10{,}000$ draws per row.  }
    \label{fig:HypocycloidNearCusp}
\end{figure}

\begin{figure}[p]
    \centering
    \includegraphics[width=0.92\linewidth]{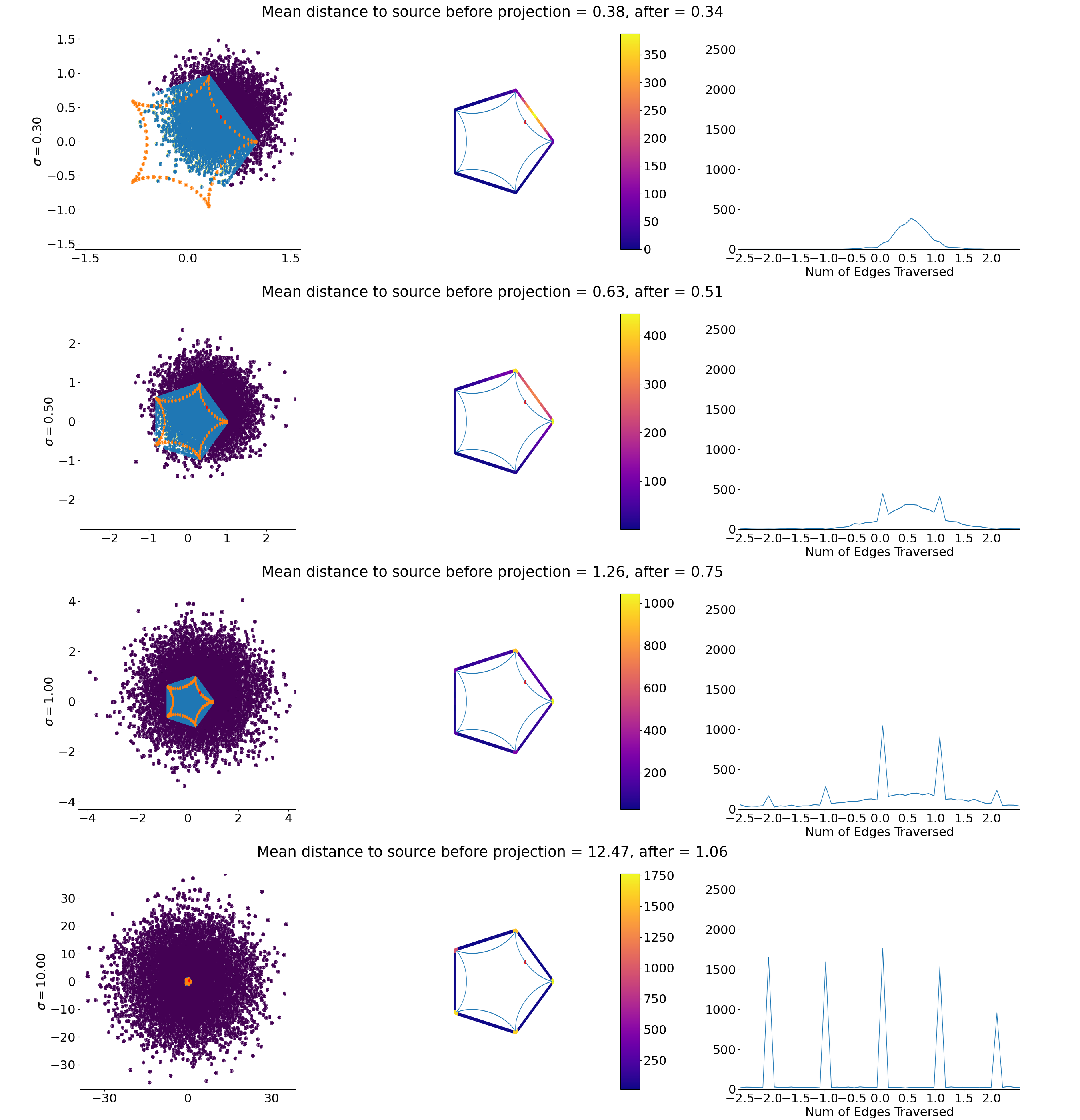}
    \caption{Hypocycloid experiment with the source point at the midpoint of
    an arc.  Rows correspond to $\sigma=0.3,0.5,1,$ and $10$, from top to
    bottom, with $10{,}000$ draws per row. }
    \label{fig:HypocycloidMidArc}
\end{figure}

\clearpage

In Figure \ref{fig:HypocycloidCusp}, we can see for $\sigma=0.3$ the projected points concentrate at the source vertex, for $\sigma=0.5$ they become less concentrated, and for $\sigma=1.0$ concentrations begin to appear at the two adjacent vertices. Finally, $\sigma=10.0$ has most of the exterior points projecting to the vertices, although the source vertex and then adjacent vertices are more heavily concentrated.

Figure \ref{fig:HypocycloidNearCusp} tells a similar story, although it is clear from the first few rows that the concentration is slightly off center from the vertex closest to the source.

Finally, in Figure~\ref{fig:HypocycloidMidArc}, for $\sigma=0.3$ most
of the exterior points project to the interior of the edge closest to the
source. For $\sigma=0.5$ the distribution extends to the adjacent vertices.
For $\sigma=1.0$, the two adjacent vertices produce the strongest peaks,
and for $\sigma=10.0$ the vertex masses are more evenly distributed.

\section{Estimating directional support values from Gaussian tails}\lab{sec:HypDist}

\begin{figure}[H]
    \centering
        \includegraphics[width=2.5in]{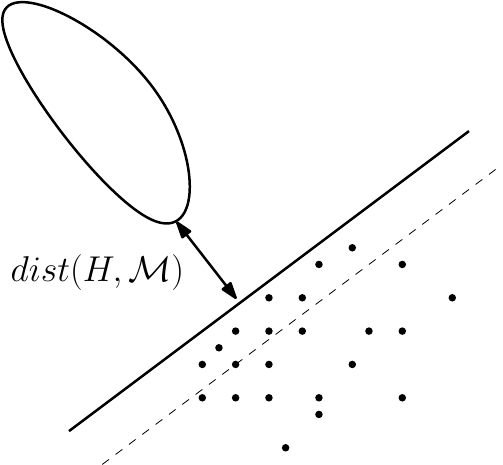} 
        \caption{Estimating the distance of $H$ from $\MM$ by counting the number of samples on the side of $H$ that does not contain $\MM$.} \label{fig:distHM}
    \end{figure}

In this section $H$ will always denote a hyperplane contained in $\R^D.$

For a unit vector $u\in S^{D-1}$, write
\[
 h(u):=h_K(u)=h_{\MM}(u)=\sup_{x\in\MM}\langle x,u\rangle.
\]
For every $\gamma\ge h(u)$, the hyperplane
$H_{\gamma,u}=\{y:\langle y,u\rangle=\gamma\}$ supports $K$ or lies beyond it in direction
$u$, and
\begin{equation}\label{eq:hyperplane-distance-support}
 \dist(H_{\gamma,u},\MM)=\dist(H_{\gamma,u},K)=\gamma-h(u).
\end{equation}
Thus estimating the location at which the Gaussian hyperplane density reaches
a prescribed small value, followed by subtracting the corresponding offset,
is a method for estimating the signed support value $h(u)$. The restriction
$\gamma\ge h(u)$ is essential: ordinary distance becomes zero or loses its
sign when a halfspace intersects $K$.


Recall that $\MM \subseteq B_D(0,1)$. Here we bound the hyperplane density
in terms of $\dist(H,\MM)$ beyond the support. Appendix~\ref{app:support-oracle}
uses these bounds to analyze the signed support estimate computed by Algorithm~3.

The remainder of this section closely follows \cite[Section 8]{large_noise}.
We use $\varrho$ to denote the density given by
\[
\varrho(y):= \int_\MM (\sqrt{2\pi}\sigma)^{-D}\exp\left(- \frac{ |y - x|^2}{2\sigma^2}\right)\mu_\MM(dx),
\] where $y \in \R^D$ and $\mu_\MM$ is the push forward of $\mu_0$ to $\MM$ via the projection map on $\mathfrak{S}$.
Thus, for every $x\in\MM$ and $0<\eps\le2$, the projected
measure satisfies
\[
 \mu_\MM(B(x,\eps))\ge c_{\MM_0}\omega_d\eps^d.
\]
Indeed, choose a preimage $x^*\in\MM_0$ of $x$.
Since $\Pi_{\mathfrak S}$ is $1$-Lipschitz,
$\Pi_{\mathfrak S}^{-1}B(x,\eps)$ contains $B(x^*,\eps)$. Hence
\begin{align*}
 \mu_\MM(B(x,\eps))
 &=\mu_0\bigl(\Pi_{\mathfrak S}^{-1}B(x,\eps)\bigr)\\
 &\ge\mu_0(B(x^*,\eps))
 \ge c_{\MM_0}\omega_d\eps^d
\end{align*}
by~\eqref{eq:basic}. Write $\dist(H,x)$ for the Euclidean distance
from $x$ to $H$.
Let us denote by $\dist(H, \MM)$, the $\ell_2$ distance between the two nearest points  $x$ and $y,$ where $x\in \MM$ and $y \in H$.
Let \begin{align*} \Gamma(H):= \int_H \varrho(y)\la^{D-1}_H(dy)\end{align*}
and
 \[
H_{{{t}}, u} = \{ { y\in \R^D\,}|\, \langle u,y\rangle = {{t}}\}, 
\]
where $u\in \R^D$, $\|u\|_2 = 1$,  and $\lambda_H^{D-1}$ denotes the Lebesgue measure. 

For fixed $u$ and a projected noisy observation $Y$, define the
survival probability
\[
 F_\gamma:=\mathbb P\{\langle Y,u\rangle>\gamma\}.
\]

\begin{observation} Note that assuming that for any $\gamma' \geq \gamma$, $K$ does not intersect $H_{\gamma', u}$, we have that  \begin{align} \Gamma(H_{\gamma + \de, u}) \leq \de^{-1}(F_\gamma - F_{\gamma + \de}) \leq \Gamma(H_{\gamma, u}).\end{align} Let the set of all such $\gamma$ be denoted $W_u$.
\end{observation}
For $0<\delta\le2$, recall
\begin{align} 
& &\kappa_0 := \sqrt{2\pi}\sigma,\lab{eq:k0}\\
 & &\kappa_1 := \frac{\sqrt{2\pi}\sigma}{c_{\MM_0}\de^d \omega_d }.\lab{eq:k1}
 \end{align} 

\begin{lemma}\lab{lem:2.1-12oct}
The function $\gamma\to \Gamma(H_{\gamma, u})$ is  strictly decreasing for $\gamma \in W_u$. Suppose $H$ is chosen such that $\Gamma(H)\max(\kappa_0, \kappa_1) \leq 1$ \textcolor{black}{(equivalently, $\Gamma(H)\kappa_1\le 1$, since $\kappa_1\ge\kappa_0$ whenever $c_{\MM_0}^{-1}\delta^{-d}\omega_d^{-1}\ge 1$, which we assume throughout)}.
Moreover, { the distance of the manifold $\MM$ and the affine hyperplane $H=H_{\gamma, u}$, $\gamma \in W_u$, satisfies}
\begin{align*} - {\delta} +\sqrt{ (2 \sigma^2)\log \left(({\Gamma(H)}{\kappa_1})^{-1}\right)}  \leq    \dist(H, \MM) \leq  \sqrt{(2 \sigma^2)\log \left(({\Gamma(H)}{\kappa_0})^{-1}\right)}.\end{align*}
\end{lemma}
\begin{proof}
We integrate along  $H$, and see that
\begin{align}\label{formula for Gamma}
\Gamma(H) =  \int_H \varrho(y)\la^{D-1}_H(dy) = \int_\MM (\sqrt{2\pi}\sigma)^{-1}\exp\left(- \frac{ \dist(x,H)^2}{2\sigma^2}\right)
 {\mu_\MM(dx)},
\end{align} 
where  $\la^{D-1}_H$ is the $D-1$ dimensional Lebesgue measure on $H$.
As $\MM\subset B_D(0, 1)$, for all $x\in \MM$ the function $\gamma\to \dist(x,H_{\gamma, u})$ is strictly increasing for $\gamma \in W_u$. This implies that the function $\gamma\to \Gamma(H_{\gamma, u})$ is a strictly decreasing for $\gamma \in W_u$.

We observe that 
\begin{align*}
\int_\MM (\sqrt{2\pi}\sigma)^{-1}\exp\left(- \frac{ \dist(H, x)^2}{2\sigma^2}\right){\mu_\MM(dx)} \end{align*}
\begin{align} \leq  \int_\MM (\sqrt{2\pi}\sigma)^{-1}\exp\left(- \frac{ \dist(H, \MM)^2}{2\sigma^2}\right){\mu_\MM(dx)}
,\lab{eq:7}\end{align} 
and therefore, that 
\begin{align*} \exp\left(- \frac{ \dist(H, \MM)^2}{2\sigma^2}\right) & \geq & \frac {\int_\MM (\sqrt{2\pi}\sigma)^{-1}\exp\left(- \frac{ \dist(H, x)^2}{2\sigma^2}\right)\mu_\MM(dx)}{(\sqrt{2\pi}\sigma)^{-1}} =  \frac{\Gamma(H)}{(\sqrt{2\pi}\sigma)^{-1}} .\end{align*}
We then see that  \begin{align} \dist(H, \MM) \leq  \sqrt{(- 2 \sigma^2)\log \left({\Gamma(H)}{\kappa_0}\right)}.\lab{eq:10}\end{align}

Let $\MM_\delta$ denote the set of points in $\MM$ whose distance from $H$ is less or equal to $\dist(H, \MM) + \delta,$ and let \textcolor{black}{$x^\star$} be a nearest point on $\MM$ to $H$. We see that $ \MM_{\delta}$ contains the image under $\Pi_{\mathfrak S}$ of the intersection of \textcolor{black}{the ambient Euclidean $n$-ball of radius $\delta$ around a preimage of $x^\star$} with $\MM_0$.
Similarly to (\ref{eq:7}), 
\begin{align*} 
\Gamma(H) &=& \int_\MM (\sqrt{2\pi}\sigma)^{-1}\exp\left(- \frac{ \dist(H, x)^2}{2\sigma^2}\right)\mu_\MM(dx) 
\\& \geq &   \int_{\MM_{\delta}} (\sqrt{2\pi}\sigma)^{-1}\exp\left(- \frac{ ({\delta} + \dist(H, \MM))^2}{2\sigma^2}\right)\mu_\MM(dx)
\\
& \geq & c_{\MM_0} \de^d 
 \omega_d  (\sqrt{2\pi}\sigma)^{-1}\exp\left(- \frac{ ({\delta} + \dist(H, \MM))^2}{2\sigma^2}\right).\end{align*} 
It follows that 
\begin{align*}
 \exp\left(- \frac{ ({\delta} + \dist(H, \MM))^2}{2\sigma^2}\right) & \leq &
 \frac{{ c_{\MM_0}^{-1}}}{\delta^d \omega_d }
  \frac{1}{(\sqrt{2\pi}\sigma)^{-1}}\Gamma(H)=\kappa_1\Gamma(H).
  \end{align*}
We then see that  \begin{align} ({\delta} + \dist(H, \MM)) \geq \sqrt{ (- 2 \sigma^2)\log \left({\Gamma(H)}{\kappa_1}\right)},\lab{eq:52}\end{align} and putting this together with (\ref{eq:10}), we obtain 
 \begin{align} - {\delta} +\sqrt{ (2 \sigma^2)\log \left(({\Gamma(H)}{\kappa_1})^{-1}\right)}  \leq    \dist(H, \MM) \leq  \sqrt{(2 \sigma^2)\log \left(({\Gamma(H)}{\kappa_0})^{-1}\right)}\lab{eq:53.1}.\end{align}
 This proves Lemma \ref{lem:2.1-12oct}.
 \end{proof}

The empirical support guarantee, including its parameter conditions and
sample bound, is established in Appendix~\ref{app:support-oracle}.
Appendix~\ref{sec:5} then converts the uniform support error into a
bound on the projection error.

    \section{Correctness of the algorithm}\lab{sec:6}

We now place the convex-projection argument on a single, dimensionless
scale.  Fix \(x\in\MM\), let \(K=\conv(\MM)\), and write
\[
 Z\sim\NN(0,\sigma^2I_D),\qquad Y=x+Z,\qquad
 \widehat X=P_K(Y).
\]
Define
\[
 K_\sigma:=\sigma^{-1}K,\qquad x_\sigma:=\sigma^{-1}x,\qquad
 G:=\sigma^{-1}Z\sim\NN(0,I_D).
\]
Scaling commutes with Euclidean projection, so
\[
 \widehat x_\sigma:=P_{K_\sigma}(x_\sigma+G)
       =\sigma^{-1}\widehat X.
\]
For \(t\geq0\), define the dimensionless function
\begin{equation}\label{eq:function_chatterjee}
 f_{x,\sigma}(t):=
 \E\!\left[
   \sup_{\substack{u\in K_\sigma\\\|u-x_\sigma\|_2\leq t}}
   \langle G,u-x_\sigma\rangle
 \right]-\frac{t^2}{2},
\end{equation}
and denote its unique maximizer by \(t_{x,\sigma}\).

\begin{theorem}[Chatterjee]\label{thm:chatterjee}
If \(t_{x,\sigma}>0\), then for every \(\gamma\geq0\),
\begin{equation}\label{eq:chatterjee_inequality}
 \p\!\left\{
 \left|\|\widehat x_\sigma-x_\sigma\|_2-t_{x,\sigma}\right|
 \geq\gamma\sqrt{t_{x,\sigma}}
 \right\}
 \leq
 3\exp\!\left\{
 -\frac{\gamma^4}
 {32\left(1+\gamma/\sqrt{t_{x,\sigma}}\right)^2}
 \right\}.
\end{equation}
Equivalently, in the original units,
\begin{equation}\label{eq:chatterjee_inequality-2}
 \p\!\left\{
 \left|\|\widehat X-x\|_2-\sigma t_{x,\sigma}\right|
 \geq\sigma\gamma\sqrt{t_{x,\sigma}}
 \right\}
 \leq
 3\exp\!\left\{
 -\frac{\gamma^4}
 {32\left(1+\gamma/\sqrt{t_{x,\sigma}}\right)^2}
 \right\}.
\end{equation}
If \(t_{x,\sigma}=0\), then \(K_\sigma=\{x_\sigma\}\), and
the projection error is zero almost surely.
\end{theorem}

\noindent
This is Theorem~1.1 of
\cite{9a533cee-7bd2-34ce-a747-8254219de6bc}, applied to
\((K_\sigma,x_\sigma,G)\).  In particular, \(t_{x,\sigma}\) is
dimensionless.  

We next bound the unrestricted Gaussian width used to dominate the
localized supremum in~\eqref{eq:function_chatterjee}. Set
\(\MM_\sigma:=\sigma^{-1}\MM\). Since
\(K_\sigma=\conv(\MM_\sigma)\), linearity gives, for every realization
of \(G\),
\begin{equation}\label{eq:convex-hull-width}
 \max_{u\in K_\sigma}\langle G,u-x_\sigma\rangle
 =\max_{u\in\MM_\sigma}\langle G,u-x_\sigma\rangle.
\end{equation}
The maxima are attained because the projected support \(\MM_\sigma\)
and its convex hull are compact. Thus only the covering numbers of
\(\MM_\sigma\) are needed for this unrestricted width bound.
The localized supremum in~\eqref{eq:function_chatterjee} remains indexed
by \(K_\sigma\cap B(x_\sigma,t)\): intersecting a convex hull with a
ball need not preserve the equality in~\eqref{eq:convex-hull-width}.

Throughout this section, \(N_C(r,A)\) denotes the Euclidean
\(r\)-covering number of \(A\), with centers in \(A\).
\begin{lemma}\label{lemma:covering_number_support}
For \(0<\sigma\varepsilon\leq2\),
\begin{equation}\label{eq:support-cover}
 N_C(\varepsilon,\MM_\sigma)
 =N_C(\sigma\varepsilon,\MM)
 \leq\frac{1}
 {c_{\MM_0}\omega_d(\sigma\varepsilon/2)^d}.
\end{equation}
\end{lemma}

\begin{proof}
The projected measure on \(\MM\) obeys the same lower-mass bound as
\(\mu_0\), because orthogonal projection is \(1\)-Lipschitz.
For \(r=\sigma\varepsilon\), let \(A_r\subset\MM\) be a maximal
\(r\)-separated set. The open balls \(B(a,r/2)\), \(a\in A_r\),
are pairwise disjoint. Their projected masses are bounded below by
\[
 m_r:=c_{\MM_0}\omega_d(r/2)^d.
\]
Hence \(|A_r|\le m_r^{-1}\). Maximality ensures that every point of
\(\MM\) is within distance \(r\) of \(A_r\). Rescaling this cover
proves~\eqref{eq:support-cover}.
\end{proof}

\begin{lemma}\label{lemma:RDEI_convex_hull}
Assume \(0<\sigma\leq1\) and \(D\geq1\).  There is an absolute
constant \(C\) such that
\begin{equation}\label{eq:51}
\begin{aligned}
 W_\sigma(K)=W_\sigma(\MM)
 &:=
 \E\sup_{u\in\MM_\sigma}\langle G,u-x_\sigma\rangle
 \leq J,\\
 J
 &:=\frac{C}{\sigma}
 \sqrt{\log\!\left(\frac{1}{c_{\MM_0}\omega_d}\right)}.
\end{aligned}
\end{equation}
\end{lemma}

\begin{proof}
By~\eqref{eq:convex-hull-width}, we may index the Gaussian process
\(X_u=\langle G,u-x_\sigma\rangle\) by \(u\in\MM_\sigma\).
Its canonical metric is \(\|u-v\|_2\), and
\(\diam(\MM_\sigma)\leq2/\sigma\). Dudley's entropy
bound~\cite{vershynin2018high} gives
\[
 W_\sigma(K)
 \leq C\int_0^{2/\sigma}
 \sqrt{\log N_C(\varepsilon,\MM_\sigma)}\,d\varepsilon.
\]
Set \(B:=\log(1/(c_{\MM_0}\omega_d))\). The lower-mass
condition~\eqref{eq:basic}, applied at radius \(2\), implies
\(c_{\MM_0}\omega_d2^d<1\), and hence \(B>d\log2\).
Lemma~\ref{lemma:covering_number_support} yields
\[
 \log N_C(\varepsilon,\MM_\sigma)
 \leq B+d\log\!\left(\frac2{\sigma\varepsilon}\right),
 \qquad 0<\varepsilon\leq2/\sigma.
\]
This entropy bound is integrable at zero. Indeed, substituting
\(s=\sigma\varepsilon/2\) and using
\(\sqrt{a+b}\leq\sqrt a+\sqrt b\) gives
\begin{align*}
 \int_0^{2/\sigma}
 \sqrt{B+d\log\!\left(\frac2{\sigma\varepsilon}\right)}
 \,d\varepsilon
 &=\frac2\sigma\int_0^1\sqrt{B+d\log(1/s)}\,ds\\
 &\leq\frac2\sigma
 \left(\sqrt B+\frac{\sqrt{\pi d}}2\right)
 \leq\frac{C\sqrt B}{\sigma}.
\end{align*}
Thus \(W_\sigma(K)\leq C\sqrt B/\sigma=J\), proving
\eqref{eq:51}.
\end{proof}

Because \(x_\sigma\in K_\sigma\), we have
\(f_{x,\sigma}(0)=0\).  Moreover,
\[
 f_{x,\sigma}(t)
 \leq W_\sigma(K)-\frac{t^2}{2}
 \leq J-\frac{t^2}{2}.
\]
At the maximizer \(t_{x,\sigma}\), therefore,
\[
 0\leq f_{x,\sigma}(t_{x,\sigma})
 \leq J-\frac{t_{x,\sigma}^2}{2},
\]
and hence
\begin{equation}\label{eq:time_maxima_bound}
 t_{x,\sigma}\leq T:=\sqrt{2J}.
\end{equation}

If \(t_{x,\sigma}=0\), convexity forces \(K_\sigma=\{x_\sigma\}\):
otherwise a nontrivial line segment starting at \(x_\sigma\) would make
\(f_{x,\sigma}(t)>0=f_{x,\sigma}(0)\) for all sufficiently small
\(t>0\).  In that degenerate case the projection error is zero.
Otherwise, set \(u=\gamma\sqrt{t_{x,\sigma}}\) in
\eqref{eq:chatterjee_inequality}.  Its right-hand side becomes
\[
 3\exp\!\left\{
 -\frac{u^4}{32(t_{x,\sigma}+u)^2}
 \right\}.
\]
Since \(t_{x,\sigma}\leq T\), the event
\(\{\|\widehat x_\sigma-x_\sigma\|_2\geq T+u\}\) is contained in
\(\{\|\widehat x_\sigma-x_\sigma\|_2-t_{x,\sigma}\geq u\}\).
The resulting upper-tail probability is at most
\[
3\exp\!\left\{
-\frac{u^4}{32(T+u)^2}
\right\}.
\]
Consequently, taking \(u=\gamma\sqrt T=\gamma(2J)^{1/4}\) gives the
uniform risk bound
\begin{equation}\label{eq:convex-projection-risk}
 \p\!\left\{
 \|\widehat X-x\|_2
 \geq
 \sigma\left((2J)^{1/2}+\gamma(2J)^{1/4}\right)
 \right\}
 \leq
 3\exp\!\left\{
 -\frac{\gamma^4}
 {32\left(1+\gamma/(2J)^{1/4}\right)^2}
 \right\}.
\end{equation}

  \section{Main Results}\lab{sec:main}

To establish the end-to-end theoretical guarantee of our denoising framework, we must synthesize the three distinct sources of error analyzed in the main text and appendices:
\begin{enumerate}
    \item \textbf{Subspace Approximation Error (Subsection~\ref{ssec:Pi}):} The geometric distance between the true manifold $\mathcal{M}_0$ and its orthogonal projection $\Pi \mathcal{M}_0$ onto the empirical PCA subspace $S$.
    \item \textbf{Statistical Risk (Section~\ref{sec:6}):} The fundamental statistical error of the exact Euclidean projection onto the convex set $K = \text{conv}(\Pi \mathcal{M}_0)$ within the reduced subspace, governed by Chatterjee's risk bounds.
    \item \textbf{Algorithmic Optimization Error (Appendix~\ref{sec:5}):} The algorithmic error introduced by utilizing the finite-sample support-function oracle to compute the projection onto the convex hull, rather than an exact analytical projection.
\end{enumerate}

We now combine Proposition~\ref{prop:1} with the Gaussian projection
risk bound. Theorem~\ref{thm:uniform-support-oracle} and
Proposition~\ref{prop:support-to-projection} control the remaining
algorithmic error, giving the following end-to-end result.

\begin{theorem}\lab{thm:main}
Let $\mu_0$ be a probability measure with unknown compact support
$\MM_0=\supp\mu_0\subset B_n(0,1)$, satisfying the lower-mass
condition~\eqref{eq:basic} for a known $d\in\Z_+$ and a supplied valid
constant $c_{\MM_0}>0$. Assume that the noise level $0<\sigma\le1$ is
known. Let $X\sim\mu_0$ be a fresh clean latent sample and let
$Y=X+Z$, where $Z\sim\NN(0,\sigma^2I_n)$ is independent of $X$.
Use independent PCA and oracle batches drawn from the same observation
law, both independent of $(X,Z)$.

Suppose we execute Algorithm~1, utilizing $N_0$ samples to compute a
$D$-dimensional linear PCA subspace $S$ (where
$D = \min(n, \lceil c_{\MM_0}^{-1}\omega_d^{-1}\epsilon_0^{-d}\rceil)$),
and $N_{oracle}$ samples to operate the support-function oracle (Algorithm 3).

For a prescribed algorithmic tolerance $\epsilon>0$, failure probabilities
$\alpha,\eta,\beta\in(0,1)$, and risk parameter $\gamma>0$, set
\[
 L_\beta:=2+\sigma\left(\sqrt D+\sqrt{2\log(1/\beta)}\right).
\]
Let $\mathcal V=\{0\}\cup\mathcal V_1$, where $\mathcal V_1$ is a
fixed finite $\rho$-net of $S^{D-1}$ in PCA coordinates, chosen before
observing the oracle batch. Choose the common localization radius and bin
width $\delta$ to satisfy~\eqref{eq:oracle-parameter-regime}, and set
\[
 a_{\mathrm{sup}}:=3\delta.
\]
Choose $0<\rho\le1$ so that
\begin{equation}\label{eq:main-projection-parameters}
 L_\beta\rho+a_{\mathrm{sup}}+
 \sqrt{2L_\beta(L_\beta\rho+2a_{\mathrm{sup}})}\le\epsilon.
\end{equation}
Assume that
\begin{equation}\label{eq:main-support-event}
 N_{\mathrm{oracle}}\ge
 \frac{C}{\delta\Gamma_\delta}
 \log\!\left(\frac{C|\mathcal V_1|(r_\delta+3)}{\delta\eta}\right).
\end{equation}
Then with probability at least
$1-C\alpha-3\exp\left(-\frac{\gamma^4}
{32(1+\gamma/(2J)^{1/4})^2}\right)
-\eta-\beta$, the estimated denoised point $\hat{X}_{algo}$ produced by
Algorithm~1 satisfies
\begin{equation} \label{eq:main_bound}
    \|\hat{X}_{algo} - X\|_2 \le \underbrace{2\left(\frac{4\epsilon_0^2 + 2\epsilon_{emp}}{c_{\MM_0}\omega_d}\right)^{\frac{1}{d+2}}}_{\text{PCA Subspace Bias}} + \underbrace{\sigma\left((2J)^{1/2}+\gamma(2J)^{1/4}\right)}_{\text{Statistical Risk on } K} + \underbrace{\epsilon}_{\text{Algorithmic Error}}
\end{equation}
where $K = \mathrm{conv}(\Pi \mathcal{M}_0)$ is the convex hull of the projected support. Here, the empirical PCA generalization bound $\epsilon_{emp}$ is defined as:
\begin{equation}
    \epsilon_{emp} = 2(R+1)^2 \left(\frac{\sqrt{D}+2}{\sqrt{N_0}}\right)\left(1+\sqrt{2\ln(4/\alpha)}\right),
\end{equation}
with
\[
 R=C\sigma\sqrt n+C\sigma\sqrt{\log(CN_0/\alpha)},
\]
and $J$ is the following bound on the dimensionless Gaussian width of
$K_\sigma=\sigma^{-1}K$, obtained by applying Dudley's entropy integral
to the projected support $\sigma^{-1}\Pi\MM_0$ in
Lemma~\ref{lemma:RDEI_convex_hull}:
\begin{equation}
     J=\frac{C}{\sigma}
     \sqrt{\log\!\left(\frac{1}{c_{\MM_0}\omega_d}\right)}.
\end{equation}
\end{theorem}

\begin{proof}
Let $\mathcal D_0=(Y_1,\ldots,Y_{N_0})$ denote the PCA batch, and
let $\Pi$ be the orthogonal projection onto its $D$-dimensional PCA
subspace $S$. Conditional on $\mathcal D_0$, the subspace $S$ and
$K=\conv(\Pi\MM_0)$ are fixed. Identify $S$ with $\R^D$ by an
orthonormal basis. Set $\widetilde X=\Pi X$ and
$\widetilde Y=\Pi Y=\widetilde X+\widetilde Z$. Independence of the
batches and target implies that, conditional on $\mathcal D_0$,
$\widetilde Z\sim\NN(0,\sigma^2I_D)$ is independent of
$\widetilde X$, and the projected oracle observations are independent
samples from the corresponding Gaussian convolution. Let
$\widehat X_{\mathrm{exact}}=P_K(\widetilde Y)$.

By the triangle inequality, the total $\ell_2$ recovery error can be decoupled as:
\begin{equation}
    \|\hat{X}_{algo} - X\|_2 \le \|\Pi X - X\|_2 + \|P_K(\Pi Y) - \Pi X\|_2 + \|\hat{X}_{algo} - P_K(\Pi Y)\|_2.
\end{equation}

\textbf{1. Bounding the PCA Subspace Bias:} 
By Proposition~\ref{prop:1}, with probability at least $1-C\alpha$,
\[
 \sup_{x\in\mathcal M_0}\|x-\Pi x\|_2
 \le2\left(\frac{4\epsilon_0^2+2\epsilon_{emp}}
 {c_{\MM_0}\omega_d}\right)^{1/(d+2)}.
\]
Thus the distance from the unprojected sample $X$ to $\Pi X$ is bounded by
this term.

\textbf{2. Bounding the Statistical Risk of Convex Projection:} 
Conditional on $\mathcal D_0$ and $X$, the clean point $\Pi X$ lies
in the fixed set $K$, and the projected noise is isotropic Gaussian
with variance $\sigma^2$. We apply Chatterjee's theorem
(Theorem~\ref{thm:chatterjee}) to $P_K(\Pi Y)$. By Lemma~\ref{lemma:RDEI_convex_hull}, applied to $\MM=\Pi\MM_0$, the dimensionless Gaussian width is bounded by $J$.  Equation~\eqref{eq:time_maxima_bound} gives \(t_{\Pi X,\sigma}\leq\sqrt{2J}\), and the uniform upper-tail bound in \eqref{eq:convex-projection-risk},
followed by averaging over $X$, yields
\begin{equation}
\begin{aligned}
 &\mathbb P\!\left\{
 \|P_K(\Pi Y)-\Pi X\|_2
 \geq\sigma\left((2J)^{1/2}+\gamma(2J)^{1/4}\right)
 \,\middle|\,\mathcal D_0\right\}\\
 &\qquad\leq
 3\exp\!\left\{
 -\frac{\gamma^4}
 {32(1+\gamma/(2J)^{1/4})^2}
 \right\}.
\end{aligned}
\end{equation}

\textbf{3. Bounding the Algorithmic Error:}
By Theorem~\ref{thm:uniform-support-oracle}, condition
\eqref{eq:main-support-event} implies that, conditional on
$\mathcal D_0$, with probability at least $1-\eta$ the support oracle
satisfies
\[
 \max_{v\in\mathcal V_1}|\widehat h(v)-h_K(v)|\le a_{\mathrm{sup}}.
\]
Since $\Pi$ is linear, $\|\Pi X\|_2\le1$ and
$\Pi Z\sim\mathcal N(0,\sigma^2I_D)$ conditional on
$\mathcal D_0$. Gaussian norm concentration therefore
gives
\[
 \mathbb P\{\|\Pi Y\|_2+1>L_\beta\mid\mathcal D_0\}\le\beta.
\]
On the intersection of the good norm event
$\{\|\Pi Y\|_2+1\le L_\beta\}$ with the support-oracle event,
Proposition~\ref{prop:support-to-projection} and
condition~\eqref{eq:main-projection-parameters} give
\[
 \|\hat X_{algo}-P_K(\Pi Y)\|_2\le\epsilon.
\]
This argument uses signed support values and remains valid when $\Pi Y\in K$;
in that case Proposition~\ref{prop:support-to-projection} gives the sharper
algorithmic error bound $a_{\mathrm{sup}}$.

The conditional bounds in Steps 2 and 3 hold uniformly in the PCA
batch. Averaging over $\mathcal D_0$ and combining these bounds with
the PCA event in Step 1 by a union bound proves the stated probability.
\end{proof}

\section{Application to Cryo-Electron Microscopy}\lab{sec:Cryo-EM}

Let $G=SO(k)$ for some fixed $k\geq2$, with the chordal
Hilbert--Schmidt metric induced by its embedding in $\R^{k\times k}$.
Write $d_{\mathrm{geo}}$ for the geodesic distance of the induced
Riemannian metric, and let $\mu_G$ be normalized Haar measure.
Let $\Lambda:G\to\R^n$ be $1$-Lipschitz for the chordal metric, and
let $\Pi$ be an orthogonal projection onto a $D$-dimensional subspace,
identified isometrically with $\R^D$, where $D$ is as in
Section~\ref{sec:Algo}.  The image set
$\MM=\Pi(\Lambda(G))$ carries the pushforward measure
$(\Pi\circ\Lambda)_\#\mu_G$.

\subsection{Covering number of $SO(k)$}
Being a compact, $C^{1, 1}$-submanifold of $\R^{k\times k}$, $SO(k)$ has positive reach.
Let $\tau>0$ be a lower bound for this reach.
Let $\mathbf d$ denote the metric on $G=SO(k)$ induced by its standard
embedding in the Euclidean space $\R^{k\times k}$.
We use the separated-set convention: $N_P(r,A,\mathbf d)$ is the
largest cardinality of a subset of $A$ whose distinct points have
distance greater than $r$.  Put $d=\binom{k}{2}$ and
$V=\mvolume(G)$.  By the packing--covering inequality
(Lemma~4.2.6 of~\cite{vershynin2018high}) and the local volume
bound in Lemma~\ref{lem:2}, for $0<\epsilon\leq\tau/4$,
\[
 N_C(\epsilon,G,\mathbf d)
 \leq N_P(\epsilon,G,\mathbf d)
 \leq \frac{3^dV}{\omega_d\epsilon^d}
 = C_k\epsilon^{-\binom{k}{2}}.
\]
Here $C_k=3^dV/\omega_d$ depends on $k$; it is not asserted to be
an absolute constant uniformly over all rotation groups.

\subsection{A compactly supported $C^1$ density function }
The Sobolev seminorm \( \| \cdot \|_{\dot{W}^{b,p}(\Omega)} \) for a function \( u \in \dot{W}^{b,p}(\Omega) \), where \( \Omega \subset \mathbb{R}^n \), is defined in terms of the \( L^p \) norms of its weak derivatives of order \( b \).
Thus,  the $\dot{W}^{b,  p}$-Sobolev seminorm of $u$ is defined as \[
\| u \|_{\dot{W}^{b,p}(\Omega)} = \left( \sum_{|\alpha| = b} \| D^\alpha u \|_{L^p(\Omega)}^p \right)^{1/p}
\]

Let $q \in B_k(0, 1) \subset \mathbb{R}^k$ be a point inside the unit ball.  Let  $f(q)$ be a $C^1$ function supported on $B_k(0, {\frac{1}{2}})$ and let 

\[\|f\|_{\dot{W}^{1, 2}} := \|f\|_{\dot{W}^{1, 2}(B_k(0,1))}\]  denote its $\dot{W}^{1, 2}$ Sobolev seminorm.
Let $R \in \mathrm{SO}(k)$ be a rotation matrix.

Let $T(R) = f(R^{-1}q)$ denote the function obtained by applying $R$ to $f(q)$. 

Recall that we are denoting by $G$, the Lie group $SO(k)$.

Let
    \( \rho\) be a smooth action of \( G \) on \( \mathbb{R}^k \), (which in our case is given simply by matrix vector multiplication,) and let \( f\) be the \(C^1\) function above.

For \( g \in G \), define the transformed function:
\[
f^g(x) := f(\rho(g^{-1})(x)),
\]
which corresponds to the pullback of \( f \) by the group action.


Let \( X \in \mathfrak{g} \), the Lie algebra of \( G \), and let \( \gamma(t) \) be a smooth curve in \( G \) such that:
\[
\gamma(0) = e, \quad \gamma'(0) = X.
\]

Then the infinitesimal action of \( X \) on \( f \) is defined by:
\[
(X \cdot f)(x) := \left. \frac{d}{dt} \right|_{t=0} f(\rho(\exp(-tX))(x)).
\]

This defines a vector field \( X^\# \) on \( \mathbb{R}^k \), called the \emph{infinitesimal generator}, given by:
\[
X^\#(x) := \left. \frac{d}{dt} \right|_{t=0} \rho(\exp(tX))(x).
\]

Therefore,
\[
(X \cdot f)(x) = - X^\# f(x),
\]
where \( X^\# f(x) \) denotes the directional derivative of \( f \) along the vector \( X^\#(x) \), i.e.,
\[
X^\# f(x) = Df(x) \cdot X^\#(x).
\]



Let \( X \in \mathfrak{so}(k) \), a skew-symmetric matrix representing an infinitesimal rotation, such that $\|X\|_{HS} = 1$, where $\|\cdot\|_{HS}$ denotes the Hilbert-Schmidt norm on \textcolor{black}{$\R^{k\times k}$ (i.e.\ the Frobenius norm)}. Then:
\[
X^\#(x) = Xx,
\]
and the infinitesimal action becomes:
\[
(X \cdot f)(x) = - \nabla f(x) \cdot (Xx).
\]

Since $f$ is supported in $B_k(0,1/2)$, for every
$X\in\mathfrak{so}(k)$,
\[
 \|X\cdot f\|_{L^2}
 \leq\frac12\|X\|_{\op}\|\nabla f\|_{L^2}
 \leq\frac12\|X\|_{HS}\|f\|_{\dot W^{1,2}}.
\]
Rotations preserve Lebesgue measure and this support ball.  Integrating
the bound along a minimizing geodesic from $g$ to $h$ and applying
Minkowski's inequality therefore gives
\[
 \|f^g-f^h\|_{L^2}
 \leq\frac12d_{\mathrm{geo}}(g,h)\|f\|_{\dot W^{1,2}}.
\]
The geodesic and chordal metrics satisfy
$d_{\mathrm{geo}}(g,h)\leq(\pi/2)\|g-h\|_{HS}$, as verified below.
Consequently,
\begin{equation}\label{eq:T}
 \|f^g-f^h\|_{L^2}
 \leq\frac\pi4\|g-h\|_{HS}\|f\|_{\dot W^{1,2}},
 \qquad g,h\in SO(k).
\end{equation}

\begin{remark}
For the principal rotation-block angles $\theta_j\in[-\pi,\pi]$ of
$g\in SO(k)$,
\[
 d_{\mathrm{geo}}(I,g)^2=2\sum_j\theta_j^2,
 \qquad
 \|g-I\|_{HS}^2=8\sum_j\sin^2(\theta_j/2).
\]
The inequality $|\theta|\leq\pi|\sin(\theta/2)|$ on $[-\pi,\pi]$
proves the comparison at the identity.  Left invariance of both
metrics gives it for every pair $g,h$.
\end{remark}

\subsection{Operator norm of the X-ray transform}

Let $B=\{(x_1, \dots, x_{k-1}, z) \in \mathbb{R}^k | x_1^2+ \dots + x_{k-1}^2+z^2\leq 1\}$ denote the unit ball in $\mathbb{R}^k$ centered at origin and $\mathcal D=\{(x_1, \dots, x_{k-1})\in \mathbb{R}^{k-1} | x_1^2+ \dots + x_{k-1}^2 \leq 1\}$ denote the unit disc in $\mathbb{R}^{k-1}$ centered at the origin. Let $g\in L^2(B)$ and $h \in L^2(\mathcal D)$. We abbreviate  $(x_1, \dots, x_{k-1}) $ to $x$.
Extend $g$ by zero outside $B$.  Let $F$ denote the operator
corresponding to the X-ray transform given by

\begin{align}
    \textcolor{black}{h(x) \;=\; F(g)(x) \;:=\;} \int_{z=-\infty}^{\infty} g(x,z) dz 
\end{align}

Since $g(x,z)$ is restricted to the unit ball, its $z$ coordinate satisfies
\[
z^-=-\sqrt{1-|x|^2}
\qquad\text{and}\qquad
z^+=\sqrt{1-|x|^2}.
\]
Thus we have

\begin{align}
    \textcolor{black}{h(x) \;=\; F(g)(x) \;=\;} \int_{z=z^-}^{z^+} g(x, z) dz 
\end{align}

We are interested in the operator norm of $F$ defined as

\begin{align}
    \norm{F}_{\op} = \sup_{\norm{g}=1} \norm{F(g)}
\end{align}

\begin{align}
    \norm{F(g)}^2 = \int_{\mathcal D} (F(g))^2 dx = \int_{\mathcal D} \paran{\int_{z=z^-}^{z^+} g(x, z) dz}^2 dx 
\end{align}
where $dx := dx_1 \dots dx_{k-1}$.
Applying Cauchy-Schwarz inequality we have

\begin{align}
    \paran{\int_{z=z^-}^{z^+} g(x, z) dz}^2 &\leq \paran{\int_{z=z^-}^{z^+} dz} \paran{\int_{z=z^-}^{z^+} g(x, z)^2 dz} \nonumber \\
    &\leq \paran{2\sqrt{1-|x|^2}} \paran{\int_{z=z^-}^{z^+} g(x, z)^2 dz}
\end{align}

\begin{align}
    \norm{F(g)}^2 \leq \int_{\mathcal D} 2\sqrt{1-|x|^2} \paran{\int_{z=z^-}^{z^+} g(x, z)^2 dz} dx 
\end{align}

The integration domain is the unit ball $B$, so this is

\begin{align}
    \norm{F(g)}^2 \leq \int_{B} \paran{2\sqrt{1-|x|^2}} g(x, z)^2 dx dz 
\end{align}

Since $\sqrt{1-|x|^2 } \leq 1$, we have

\begin{align}
    \norm{F(g)}^2 \leq 2\int_{B} g(x, z)^2 dx dz = 2\norm{g}^2
\end{align}

Thus we have

\begin{align}\label{eq:F}
    \norm{F}_{\op} = \sup_{\norm{g}=1} \norm{F(g)} \leq \sqrt{2}
\end{align}

\subsection{Operator norm of the sampling transform}
Divide the cube $[-1/2,1/2]^{k-1}\subset\R^{k-1}$ into
$N_{pix}^{k-1}$ smaller cubes, which we call pixels. Each pixel has
$(k-1)$-dimensional volume $N_{pix}^{-(k-1)}$. We index the pixels by
\[
(i_1,\dots,i_{k-1})
\in\{0,1,\dots,N_{pix}-1\}^{k-1}.
\]

Let $h(x) \in L^2([-1/2,1/2]^{k-1})$ and \textcolor{black}{let $v=(v_i)_{i\in\{0,1,\ldots,N_{pix}-1\}^{k-1}}\in\R^{N_{pix}^{k-1}}$}. Then the sampling transform is given by $v = S(h)$ where

\begin{align}
    v_{i} = \int\limits_{i^{th} \, {\mathrm{pixel}}} N_{pix}^{\frac{k-1}{2}} h(x) dx.
\end{align}

We are interested in the operator norm of $S$ defined as

\begin{align}
    \norm{S}_{\op} = \sup_{\norm{h}=1} \norm{S(h)}
\end{align}

Consider a function $h(x)$ supported on $[-1/2, 1/2]^{k-1}$, and note that if it is a maximizer of the dilation factor $\frac{\|S(h)\|_{\ell_2}}{\|h\|_{L^2}}$, so is its absolute value $|h|(x)$, so we assume WLOG that $h$ is nonnegative. Further, within each pixel, we assume WLOG that $h$ is constant, since by Cauchy-Schwarz, this constraint does not decrease the maximum possible dilation factor. 

Thus $    h(x) = v_{i} N_{pix}^{\frac{k-1}{2}}$ if $x$ belongs to the $i^{th}$ pixel.
    
   \[ \norm{h}^2 = \int_{[-1/2,1/2]^{k-1}} h(x)^2 dx =  |v|^2. \]

Thus we have

\begin{align}\label{eq:S}
    \norm{S}_{\op} = \sup_{\norm{h}=1} \norm{S(h)} = 1.
\end{align}

\subsection{The Lipschitz constant of  $ S \circ F \circ T$}

Let $\Phi=S\circ F\circ T$ and $D_{pix}=N_{pix}^{k-1}$, extending
the output of $F$ by zero outside the unit disk and restricting it to
the pixel cube before applying $S$.
Equations~\eqref{eq:T}, \eqref{eq:F}, and \eqref{eq:S} give
\[
 \operatorname{Lip}(\Phi)
 \leq\frac{\pi\sqrt2}{4}\|f\|_{\dot W^{1,2}},
 \qquad
 \sup_{g\in SO(k)}\|\Phi(g)\|_2\leq\sqrt2\|f\|_{L^2},
\]
where the Lipschitz constant is taken for the chordal
Hilbert--Schmidt metric on $SO(k)$.  By rescaling image intensities
if necessary, we assume
\begin{equation}\label{eq:B1}
 \|f\|_{\dot W^{1,2}}\leq\frac4{\pi\sqrt2},
 \qquad
 \|f\|_{L^2}\leq\frac1{\sqrt2}.
\end{equation}
Then $\Phi$ is $1$-Lipschitz and its image lies in the unit ball of
$\R^{D_{pix}}$.  The noise standard deviation is rescaled by the
same intensity factor.

\subsection{Denoising Guarantee for Cryo-EM}

The normalization~\eqref{eq:B1} links this imaging model to the
lower-mass framework.  The domain $G=SO(k)$ is a compact smooth
manifold of dimension $d=k(k-1)/2$, with finite Riemannian volume $V$
and positive reach $\tau$ in its matrix embedding.  Since $\Phi$ is
$1$-Lipschitz,
\[
 N_C(r,\Phi(G))\leq N_C(r,G,\mathbf d),
 \qquad
 \mathcal H^d(\Phi(G))\leq\mathcal H^d(G)=V,
\]
where $\mathcal H^d$ denotes $d$-dimensional Hausdorff measure.
These bounds concern the image set and require no manifold structure
on that image.

\begin{corollary}[Finite-Sample Denoising of Cryo-EM Projections]
Assume~\eqref{eq:B1}, let $\MM_0=\Phi(SO(k))\subset B_{D_{pix}}(0,1)$,
and let $\mu_0=\Phi_\#\mu_G$, where $\mu_G$ is normalized Haar
measure.  Then $\mu_0$ satisfies~\eqref{eq:basic} with
$d=k(k-1)/2$ and a positive constant determined by $d,V,\tau$.
Suppose that $X_i$ are i.i.d.\ with law $\mu_0$ and that
$Y_i=X_i+Z_i$, where the $Z_i$ are i.i.d.\ $\NN(0,\sigma^2I_{D_{pix}})$,
independent of the clean samples, with known $0<\sigma\leq1$ in the
normalized units.  Supply Algorithm~1 with $d$ and a valid positive
lower-mass constant, and impose the sample-splitting, sample-size,
and tuning assumptions of Theorem~\ref{thm:main}.
For each fixed target independent of the PCA and oracle batches,
the denoised output satisfies the error and probability bounds of
that theorem, with ambient dimension $n=D_{pix}$.
\end{corollary}

\begin{remark}
To verify the lower-mass assertion, let $y=\Phi(g)$ and let
$B_G(g,r)=\{h\in G:\|h-g\|_{HS}<r\}$ be a chordal ball.  Then
\[
 \mu_0(B(y,r))\geq\mu_G(B_G(g,r)).
\]
The local volume bound for $G$, followed by monotonicity for larger
radii up to $2$, bounds the right-hand side below by
$c\omega_dr^d$ for some $c>0$ determined by $d,V,\tau$.
Decreasing $c$ if necessary gives the strict inequality in
\eqref{eq:basic}.  The image $\MM_0$ is compact, but need not itself
be a smooth manifold or have positive reach.  The same denoising
conclusion applies to other orientation laws whenever their image
distribution has a supplied valid lower-mass bound.
\end{remark}
   
\section{Conclusion}

We introduced an algorithmic framework for denoising Gaussian-corrupted data
sampled from a low-dimensional manifold. More generally, the algorithm applies to any clean-data
distribution satisfying the lower-mass condition~\eqref{eq:basic}. It assumes
that independent Gaussian noise of known variance is added to the clean
samples.  The observed data consist of the corrupted points; the
algorithm also requires the exponent $d$ and a valid positive lower
bound for the constant $c_{\MM_0}$ in~\eqref{eq:basic}.

For the Cryo-EM application, we bounded the Lipschitz constant of the
rotation action and the operator norms of the X-ray and pixel sampling
transforms. These bounds show that the normalized image distribution
under uniform orientations satisfies the required lower-mass condition. 

The sufficient oracle sample bound grows rapidly as the requested
localization scale decreases relative to the noise level.  Slow
minimax rates also occur in density deconvolution with supersmooth
errors such as the Gaussian; see~\cite{Fan}.  Establishing sharp
lower bounds for the support-estimation and denoising problems
studied here remains an open question.

The optimization algorithm in Appendix~\ref{sec:5} enumerates a
directional net, whose cardinality is exponential in $D$ at fixed
resolution $\rho$.  With $d,c_{\MM_0},\delta,\sigma$
and the remaining accuracy and confidence parameters fixed,
the oracle sample bound depends on the directional set through
$\log|\mathcal V_1|$, contributing at most $D\log(3/\rho)$;
its full dependence on $D$ also depends on how the tuning parameters
are chosen.  A more efficient optimization subroutine could reduce
the cost per observation even in regimes where many samples are
required.  It would be interesting to determine whether zeroth-order
methods such as those in~\cite{belloni, Shamir} can exploit the
structure of this support oracle.  Understanding the role of the
lower-mass assumption in relation to robust optimization bounds
such as~\cite{Jan} is another question for future work.

\appendix
\section{End-to-end ellipsoid experiment}\label{sec:ellipsoid-experiment}

Figure~\ref{fig:Ellipse} illustrates the observation
model and projection geometry when the convex hull is known.  This appendix tests the
complete empirical procedure, including support estimation from a separate
noisy sample. We use the uniform-angle model illustrated in that figure:
\[
 \MM_0=\{(\cos t,\tfrac12\sin t):0\leq t<2\pi\},
 \qquad Z\sim\NN(0,0.2^2I_2),
\]
with $t$ uniform on $[0,2\pi)$.  This parameterization is $1$-Lipschitz, so
for $0<\epsilon\leq2$ it gives
\[
 \mu_0(B(x,\epsilon))\geq\frac{\epsilon}{\pi}
 =c_{\MM_0}\omega_1\epsilon,
 \qquad c_{\MM_0}=\frac{1}{2\pi},\quad \omega_1=2.
\]

For each of $20$ independent repetitions, we use $N_0=1{,}000$ observations
for uncentered PCA, an independent oracle batch, and
$N_{\mathrm{test}}=5{,}000$ independent target observations.  We set $D=n=2$, so
the PCA step is an orthogonal change of coordinates and introduces no
dimension-reduction bias; this isolates the finite-sample support oracle while
still executing every stage of Algorithm~1.  The revised practical oracle is
run with
\[
 \delta_{\mathrm{loc}}=1/5,\qquad b=1/20,\qquad q_{\min}=30,
 \qquad N_{\mathrm{oracle}}=500{,}000,
\]
and $180$ equally spaced directions. We report the conservative net
parameter $\rho=2\sin(\pi/180)$ recorded with the experimental results;
the same directions have the sharper covering radius $2\sin(\pi/360)$.
The oracle receives the noisy samples, the known parameters $d$,
$c_{\MM_0}$ and $\sigma$, and the tuning parameters
$\delta_{\mathrm{loc}}$, $b$, $\rho$ and $q_{\min}$.
The equation of the ellipse is used only to compute evaluation error and
the exact-hull benchmark.
This experiment evaluates the practical algorithm; since
$\delta_{\mathrm{loc}}=\sigma$, it lies outside the convenient restricted regime
$\delta_{\mathrm{loc}}\le\sigma/2$ used in
Theorem~\ref{thm:uniform-support-oracle}.

Figure~\ref{fig:ellipsoid-end-to-end} and
Table~\ref{tab:ellipsoid-end-to-end} give the numerical results.  For these
parameters,
\[
 r_{\delta_{\mathrm{loc}}}=0.5508,\qquad
 \Gamma_{\delta_{\mathrm{loc}}}=2.86\times10^{-3},\qquad
 b\Gamma_{\delta_{\mathrm{loc}}}=1.43\times10^{-4}.
\]
Moreover, the candidate grid in~\eqref{eq:oracle-grid} satisfies
\[
 \mathcal J_{\delta_{\mathrm{loc}},b}=\{-18,\ldots,40\},\qquad
 \{jb:j\in\mathcal J_{\delta_{\mathrm{loc}},b}\}\subset[-0.9,2.0].
\]
Algorithm~3 found a nonempty qualifying bin in every one
of the $20\times180=3{,}600$ directional evaluations.  The expected oracle
count at the threshold is
$N_{\mathrm{oracle}}b\Gamma_{\delta_{\mathrm{loc}}}=71.54$, which exceeds the
practical adequacy threshold $q_{\min}=30$.  Thus the run is neither
occupancy-limited nor driven by isolated extreme observations.
After midpoint correction, the average lower and upper endpoints of the
estimated support values over the $20$ repetitions are $0.5242$ and $1.0242$;
the true directional support values lie between $0.5$ and $1$.  The mean
directional support absolute error is $0.0217$, and the mean directional
$\ell_\infty$ error is $0.0651$.

The resulting mean test error is $0.1882$, compared with $0.2503$ for the
noisy input and $0.1872$ for exact projection onto the known convex hull.
Thus the decoupled empirical oracle recovers essentially all of the improvement
available from exact convex projection.  In this example, using
$\delta_{\mathrm{loc}}=1/5$ is not itself problematic; the important changes are to
use a finer bin width and to remove the lower-edge discretization bias.

\begin{figure}[H]
    \centering
    \includegraphics[width=\linewidth]{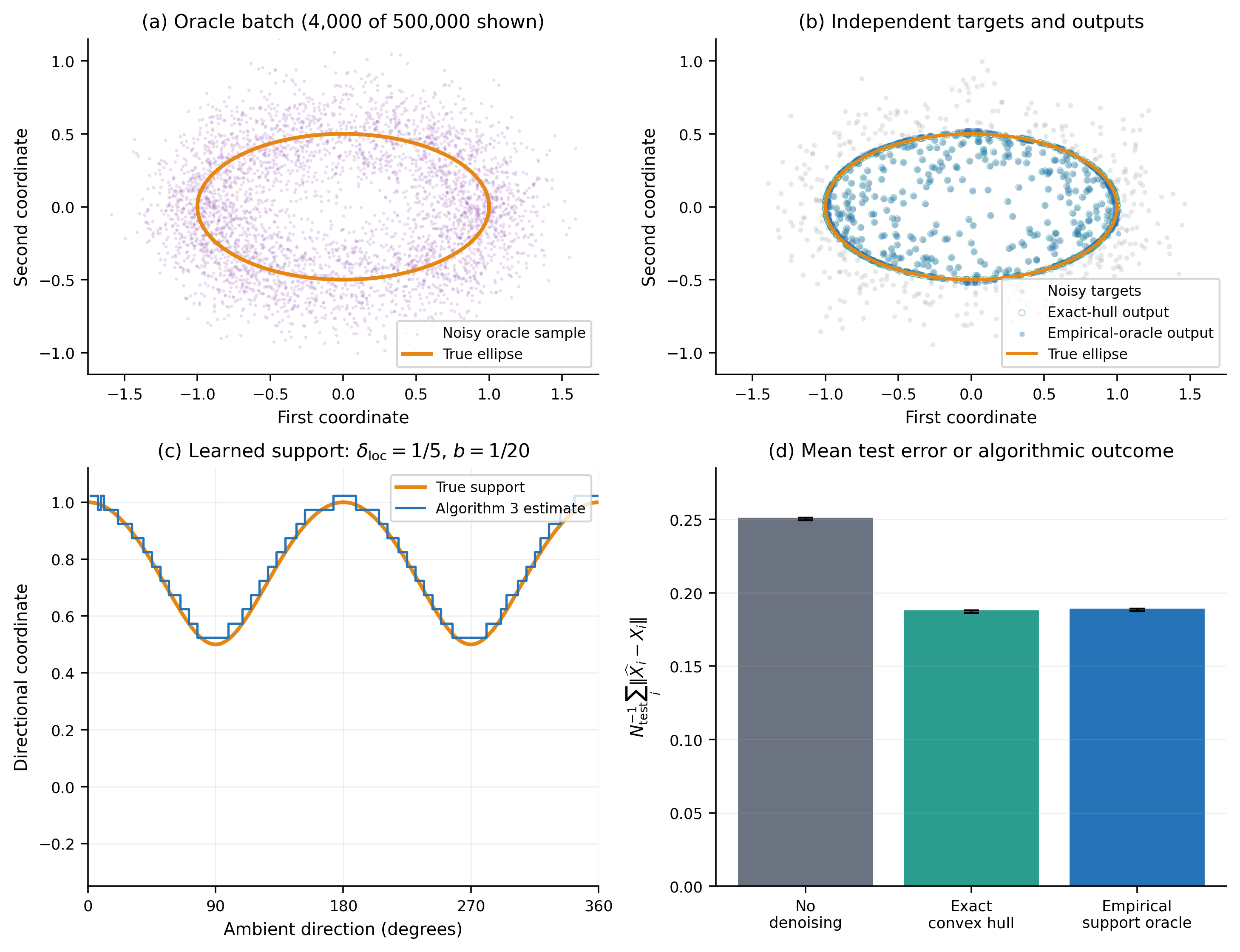}
    \caption{End-to-end experiment for the ellipsoid model in
    Figure~\ref{fig:Ellipse}.  (a) A subset of the
    noisy oracle batch.  (b) Independent noisy targets, exact-hull outputs,
    and empirical-oracle outputs.  (c) The support values estimated by
    Algorithm~3 in a representative repetition closely track the true support
    function.  (d) Mean Euclidean test errors over $20$ repetitions;
    error bars show two standard errors.}
    \label{fig:ellipsoid-end-to-end}
\end{figure}

\begin{table}[H]
\caption{End-to-end ellipsoid experiment with $\delta_{\mathrm{loc}}=1/5$,
$b=1/20$, conservative $\rho=2\sin(\pi/180)$, and
$N_{\mathrm{oracle}}=500{,}000$.
Errors are means over $20$ independent
repetitions; two standard errors are shown in parentheses.}
\label{tab:ellipsoid-end-to-end}
\centering
\small
\begin{tabular}{lccc}
\toprule
Method & Test error & Support $\ell_\infty$ error & Failed calls \\
\midrule
No denoising & $0.2503\;(0.0009)$ & -- & -- \\
Exact convex hull & $0.1872\;(0.0009)$ & -- & -- \\
Decoupled empirical support oracle & $0.1882\;(0.0009)$ & 0.0651 & $0/3{,}600$ \\
\bottomrule
\end{tabular}
\end{table}
Here “Support \(\ell_\infty\) error” is the mean, over the $20$ repetitions, of the largest absolute error among the estimated directional support values in each repetition.

\section{Miscellaneous lemmas}

\begin{lemma}\label{lem:2}
Let $\MM\subset\R^n$ be a nonempty compact $C^2$ embedded
$d$-dimensional manifold without boundary, with
$\reach(\MM)\geq\tau>0$ and $\mvolume(\MM)\leq V$.
Let $\omega_d$ be the volume of the unit ball in $\R^d$, and let
$N_P(r,\MM,\ell_2^n)$ denote the largest cardinality of a subset
whose distinct points have Euclidean distance greater than $r$.
For every $\epsilon>0$,
    \begin{align}
        N_C(3\epsilon,\MM,\ell_2^n) \leq N_P(3\epsilon,\MM,\ell_2^n) \leq \frac{V}{\omega_d \epsilon^d} \qquad &\text{if} \quad \epsilon \leq \tau/4  \nonumber \\
        N_C(3\epsilon,\MM,\ell_2^n) \leq \frac{V}{\omega_d (\tau/4)^d} \qquad &\text{if} \quad \epsilon > \tau/4
    \end{align}
\label{lemma:covering_number_manifold}
\end{lemma}

\begin{proof}
    
    A maximal $3\epsilon$-separated subset is a $3\epsilon$-net,
    so $N_C(3\epsilon,\MM,\ell_2^n)\leq N_P(3\epsilon,\MM,\ell_2^n)$
    by Lemma~4.2.6 of~\cite{vershynin2018high}.
    
    Suppose $0<\epsilon\leq\tau/4$.  Given $x\in\MM$, let
    $\Pi_x$ be orthogonal projection onto the affine tangent space
    $\Tan(x)$.  Lemma~A.1 of~\cite{FIMN}, applied with reach
    parameter $4\epsilon\leq\tau$, shows that for
    \begin{align}
        U = \{ y \in \R^n \mid \abs{y - \Pi_x y} \leq \epsilon\} \cap \{ y \in \R^n \mid \abs{x - \Pi_x y} \leq \epsilon\}
    \end{align}
    we have
    \begin{align}
        \Pi_x (U \cap \MM) = \Pi_x (U)
    \end{align}
    
    Since $U\subset B(x,3\epsilon/2)$ and orthogonal projection
    does not increase $d$-dimensional Hausdorff measure, this gives
    $\mvolume(\MM\cap B(x,3\epsilon/2))\geq\omega_d\epsilon^d$.
    Balls of radius $3\epsilon/2$ centered at the points of a
    $3\epsilon$-separated subset are disjoint.  Summing their
    intersection volumes therefore yields
    \begin{align}
        N_C(3\epsilon,\MM,\ell_2^n) \leq N_P(3\epsilon,\MM,\ell_2^n) \leq \frac{V}{\omega_d \epsilon^d} \qquad &\text{if} \quad \epsilon \leq \tau/4.
    \end{align}
    Now, consider the case when $\epsilon > \tau/4$. Since, the covering number is a monotonically decreasing function, we have for $\epsilon^\prime \geq \epsilon$ that $N_C(\epsilon^\prime,\MM,\ell_2^n) \leq N_C(\epsilon,\MM,\ell_2^n)$. Thus we have for $\epsilon > \tau/4$, 
    \begin{align}
        N_C(3\epsilon,\MM,\ell_2^n) \leq \frac{V}{\omega_d (\tau/4)^d}
    \end{align}
    
\end{proof}

\clearpage
\section{Exploratory denoising of an HK97 cryo-EM image stack}
\label{sec:hk97-exploratory}

We conclude with an exploratory application of Algorithms~1--3 to an HK97
cryo-EM particle-image stack.  This experiment is deliberately reported as a
diagnostic rather than a validated benchmark: the subset of the stack used
here contains no clean reference images, and, more importantly, the available
oracle batch does not
pass the minimum-occupancy check in Algorithm~3.  The purpose of the experiment
is to show what the PCA and forced empirical-oracle stages produce on real
cryo-EM images, while making the sample-size limitation explicit.

\subsection{Data and preprocessing}

We use a locally supplied cryo-EM stack labelled as HK97 particle data,
provided by Peter Doerschuk and John Johnson.  The precise
correspondence between this stack and the data analyzed in the
published HK97 studies has not been verified.  Veesler et
al.~\cite{VeeslerEtAl2014HK97}, Gong et
al.~\cite{GongEtAl2016HK97}, and Xu et al.~\cite{XuEtAl2018HK97}
provide related biological and methodological context; those
citations do not establish the identity of the supplied stack.

The raw MRC file contains $3{,}411$ float32 images of size $352\times352$.
Each image is treated as one high-dimensional observation.  To make the
experiment computationally manageable, every image is block-averaged over
$4\times4$ pixel blocks, producing an $88\times88$ image, and its outer-border
median is subtracted.  The first $1{,}000$ images form the uncentered PCA batch,
the next $2{,}403$ images form the oracle batch, and the eight images with
zero-based MRC indices $3{,}403$ through $3{,}410$ are held out as targets.
For each reduced dimension $D$, all projected scores are divided by the
$99$th percentile of their Euclidean norms in the PCA and oracle batches, in
order to place the projected sample on the unit-ball scale used in the
ellipsoid experiment.

\subsection{Parameters and higher-dimensional directional sets}

We retain the practical parameters of the noisy-ellipse implementation:
\[
 \sigma=0.2,\qquad \delta_{\mathrm{loc}}=\frac15,\qquad
 b=\frac1{20},\qquad q_{\min}=30,
 \qquad c_{\MM_0}=\frac1{2\pi},\qquad \omega_1=2,
\]
and use $180$ directions for each of $D=3,4,5,6$.  The phrase ``$180$
equally spaced directions'' is specific to the circle and has no direct
analogue on $S^{D-1}$ for $D>2$.  We therefore use $90$ fixed-seed random unit
directions together with their antipodes.  To describe the resulting coverage,
we report the largest nearest-direction chordal distance found among $50{,}000$
independent random probes.  This quantity is a Monte Carlo diagnostic and is
not a certified $\rho$-net radius.

For the copied ellipse calibration,
\[
 r_{\delta_{\mathrm{loc}}}=0.5508,\qquad
 \Gamma_{\delta_{\mathrm{loc}}}=2.86\times10^{-3},\qquad
 b\Gamma_{\delta_{\mathrm{loc}}}=1.43\times10^{-4}.
\]
Consequently, the available oracle batch gives
\[
 N_{\mathrm{oracle}}b\Gamma_{\delta_{\mathrm{loc}}}
 = 2{,}403(1.43\times10^{-4})=0.344<q_{\min}=30.
\]
Thus the strict algorithm returns \textsc{insufficient sample} for every value
of $D$.  In order to display the finite-stack extrapolation, the panels marked
``forced'' below bypass only this early-return check; all other listed
parameters are unchanged.  These forced outputs are not theoretically
supported executions of Algorithm~3.

\begin{table}[H]
\caption{Diagnostics for the HK97 experiment.  Correlations and relative
residual norms are computed against the noisy input and averaged over the
eight held-out targets; they are not errors against clean ground truth.}
\label{tab:hk97-d3-d6}
\centering
\small
\begin{tabular}{rrrrrr}
\toprule
$D$ & Radius diagnostic & PCA corr. & PCA residual & Forced corr. & Forced residual\\
\midrule
$3$ & $0.315$ & $0.239$ & $0.971$ & $0.193$ & $1.015$\\
$4$ & $0.604$ & $0.248$ & $0.968$ & $0.212$ & $1.017$\\
$5$ & $0.701$ & $0.264$ & $0.964$ & $0.194$ & $1.006$\\
$6$ & $0.838$ & $0.266$ & $0.963$ & $0.177$ & $1.021$\\
\bottomrule
\end{tabular}
\end{table}

Figures~\ref{fig:hk97-d3d6-a}--\ref{fig:hk97-d3d6-d} show the eight held-out
targets.  For each target, the upper row contains the raw image and the PCA
reconstructions at $D=3,4,5,6$; the lower row contains the raw image and the
corresponding forced-oracle outputs.  Every panel for a given target uses the
raw image's first-to-$99$th-percentile display window.

\begin{figure}[p]
\centering
\includegraphics[width=0.98\linewidth]{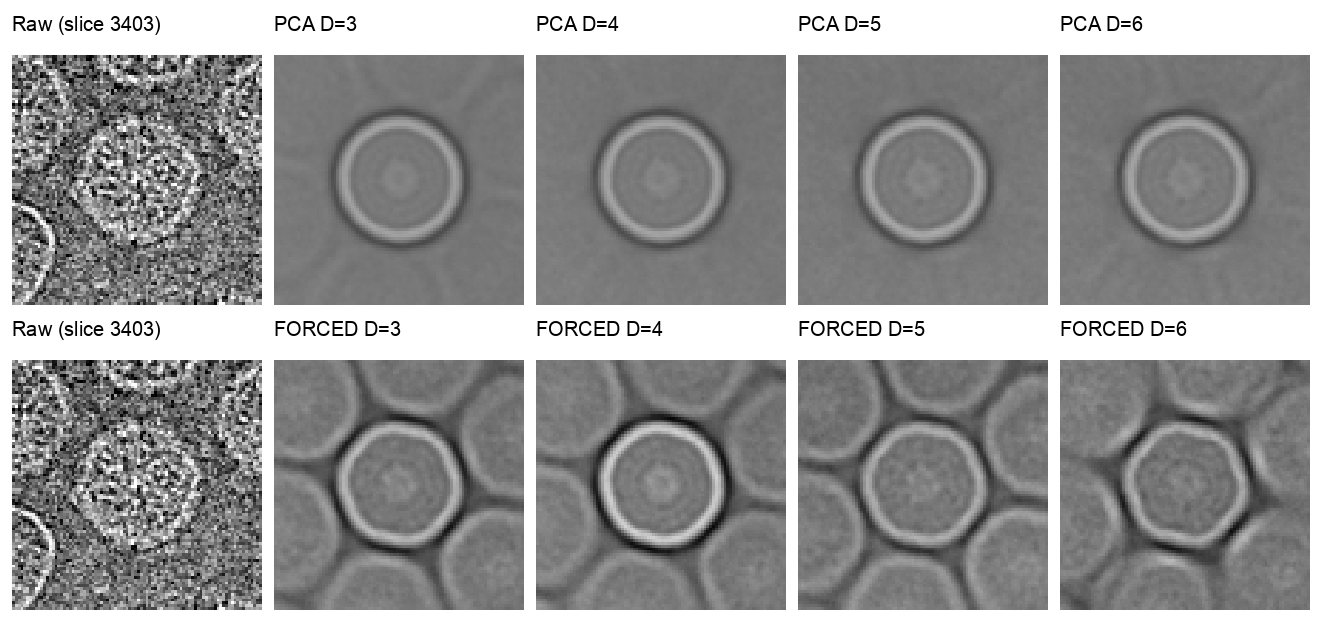}
\medskip
\includegraphics[width=0.98\linewidth]{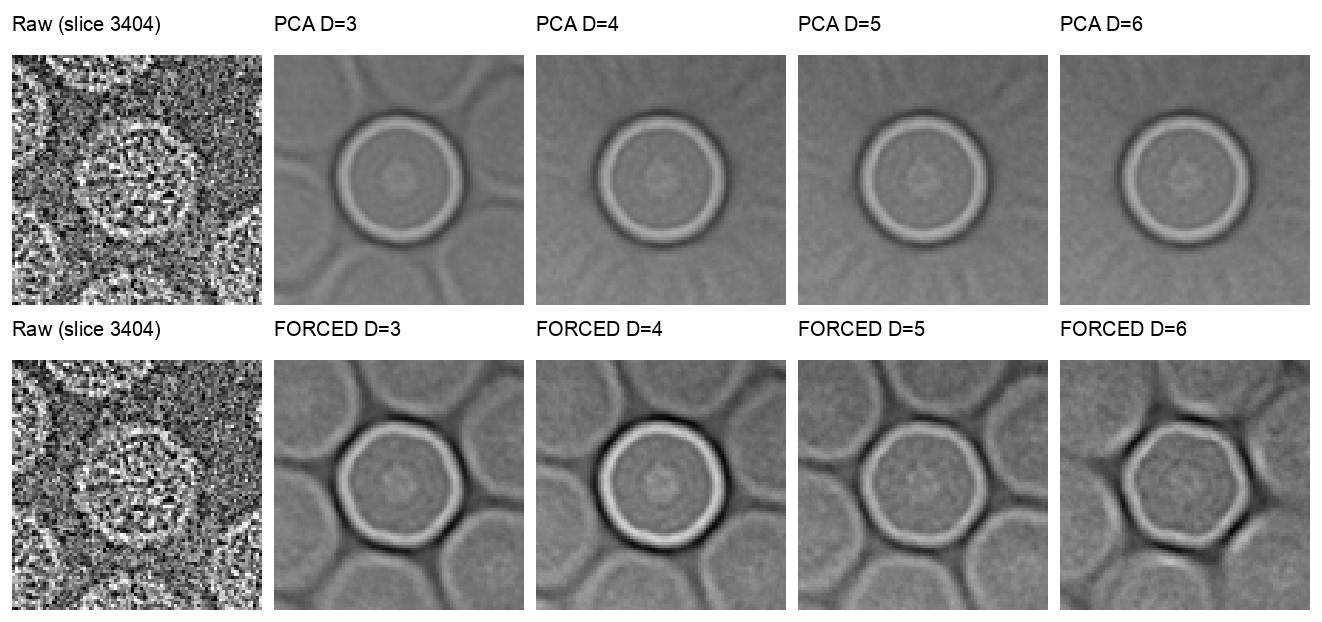}
\caption{Targets from the supplied HK97-labelled stack, with
zero-based MRC indices $3{,}403$ and $3{,}404$.
Top rows: raw images and PCA reconstructions.  Bottom rows: raw images and
forced-oracle diagnostics.}
\label{fig:hk97-d3d6-a}
\end{figure}

\begin{figure}[p]
\centering
\includegraphics[width=0.98\linewidth]{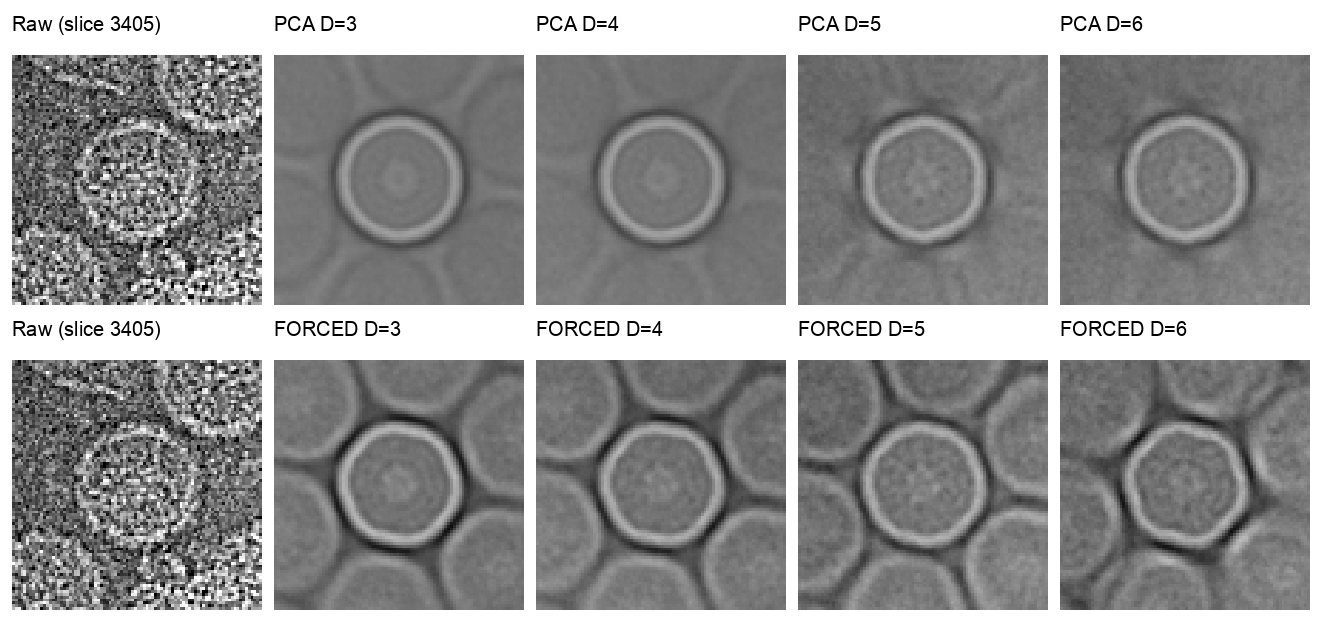}
\medskip
\includegraphics[width=0.98\linewidth]{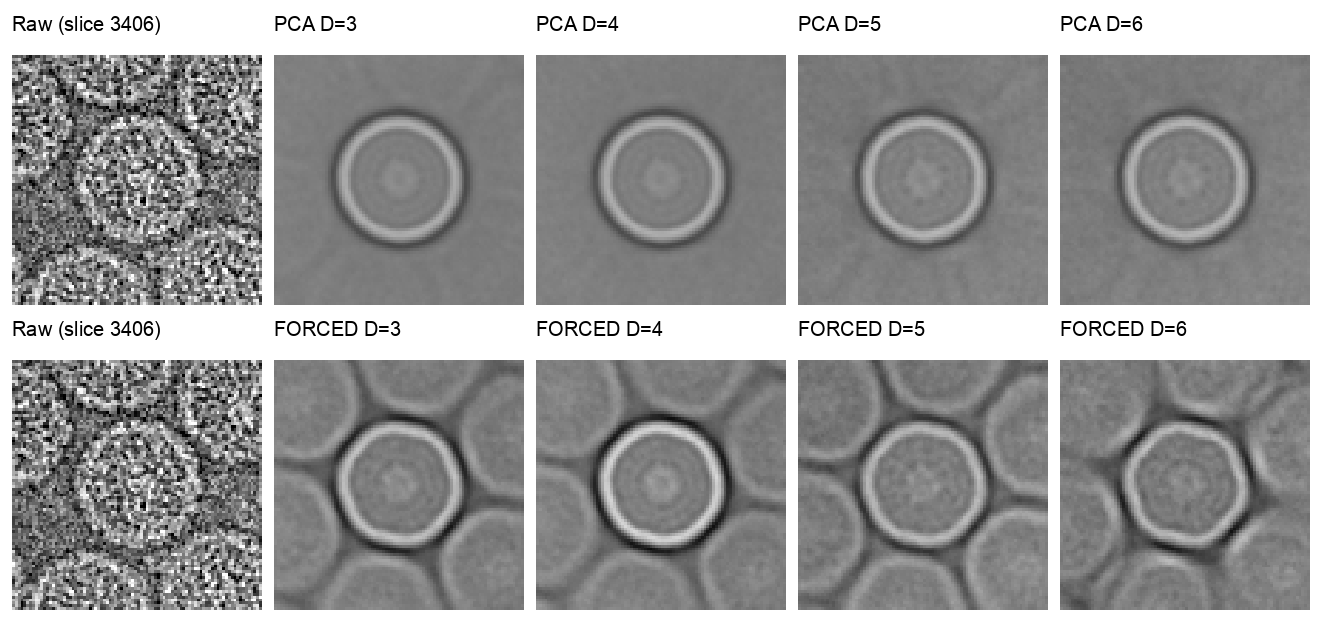}
\caption{Targets from the supplied HK97-labelled stack, with
zero-based MRC indices $3{,}405$ and $3{,}406$;
panel organization is as in Figure~\ref{fig:hk97-d3d6-a}.}
\label{fig:hk97-d3d6-b}
\end{figure}

\begin{figure}[p]
\centering
\includegraphics[width=0.98\linewidth]{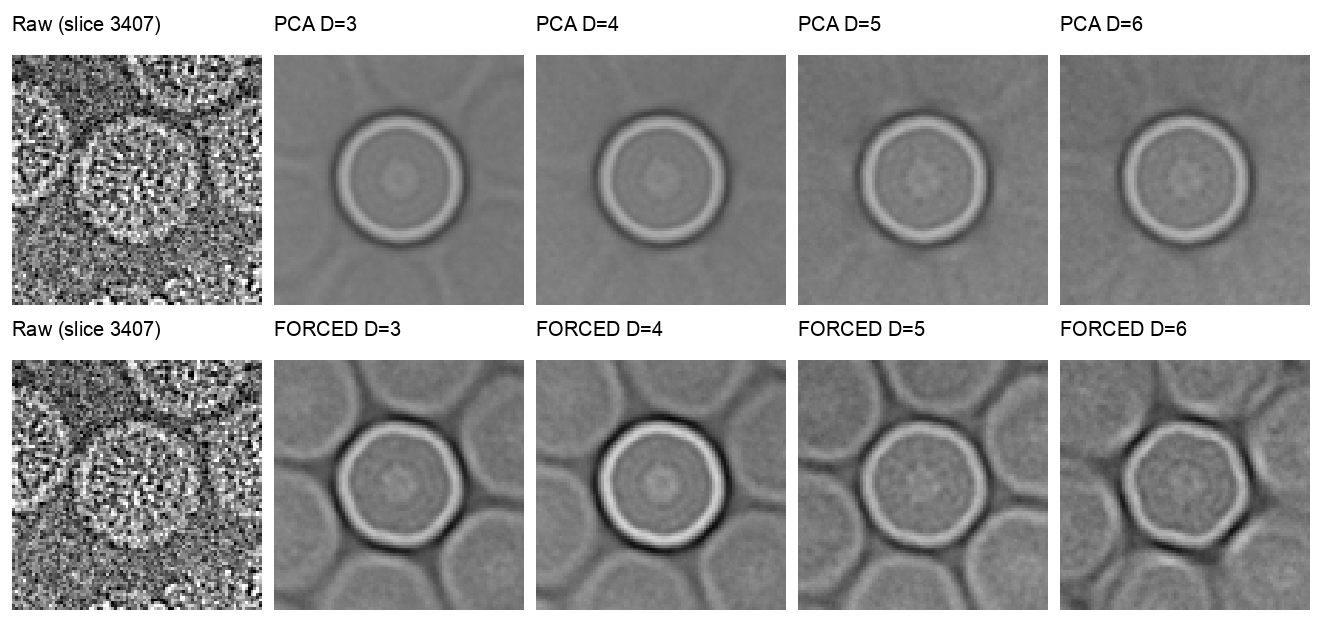}
\medskip
\includegraphics[width=0.98\linewidth]{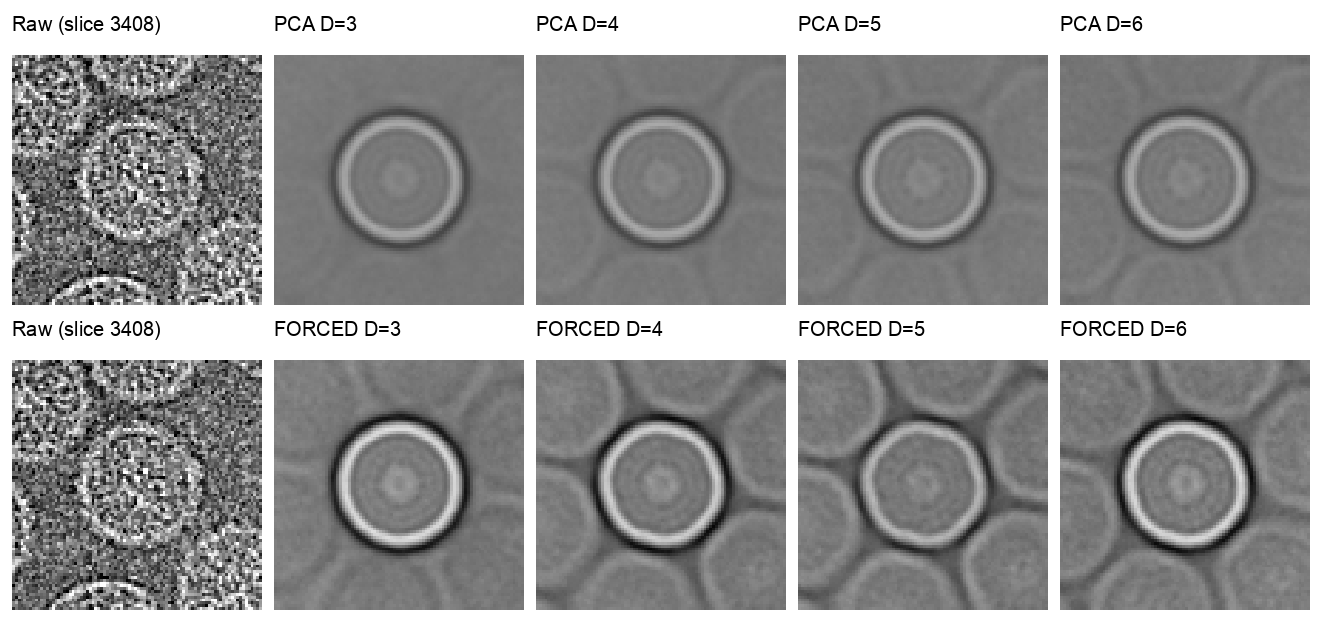}
\caption{Targets from the supplied HK97-labelled stack, with
zero-based MRC indices $3{,}407$ and $3{,}408$;
panel organization is as in Figure~\ref{fig:hk97-d3d6-a}.}
\label{fig:hk97-d3d6-c}
\end{figure}

\begin{figure}[p]
\centering
\includegraphics[width=0.98\linewidth]{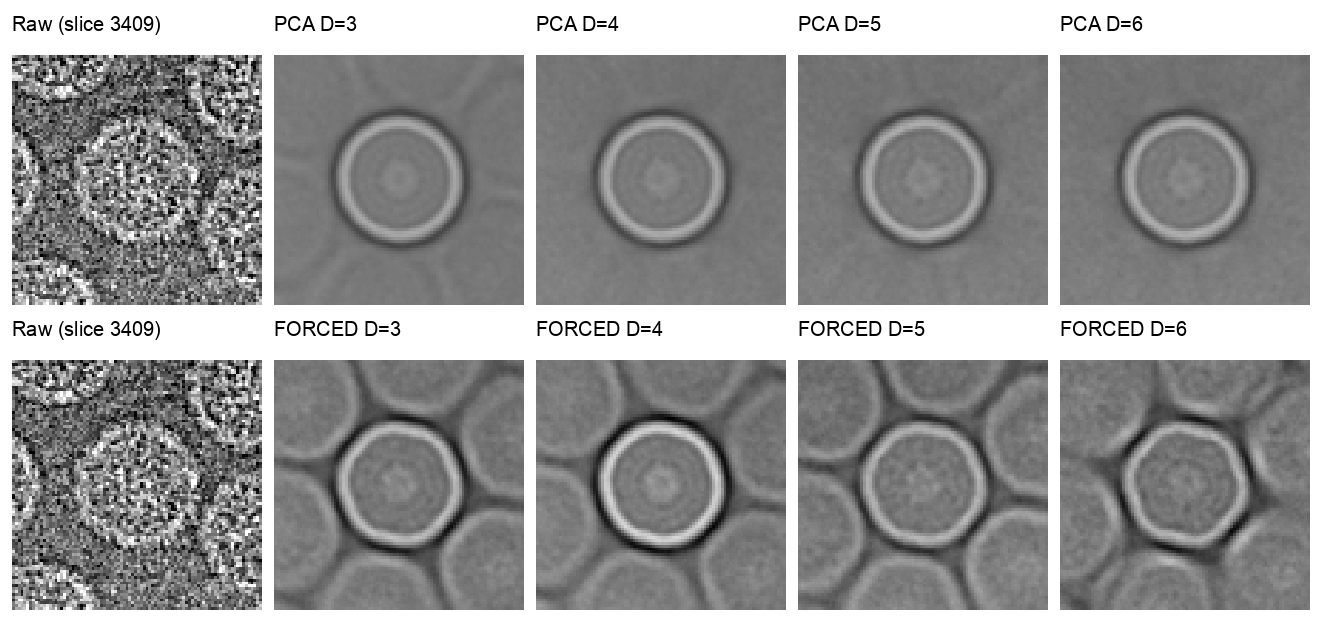}
\medskip
\includegraphics[width=0.98\linewidth]{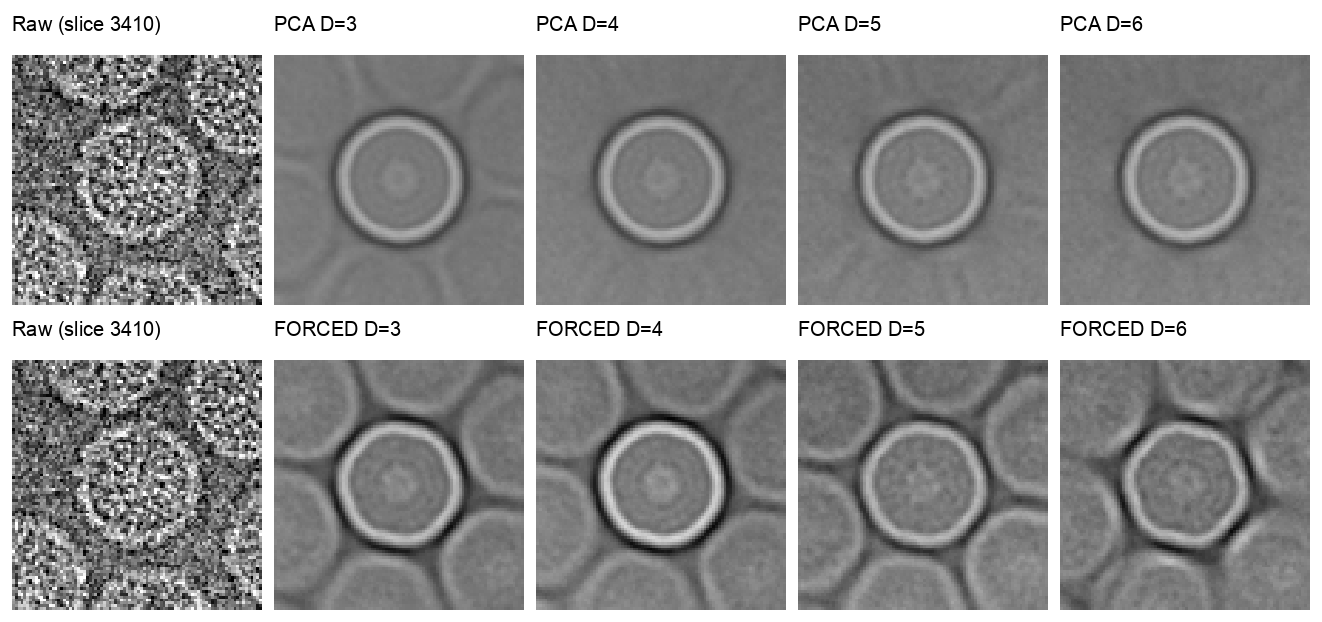}
\caption{Targets from the supplied HK97-labelled stack, with
zero-based MRC indices $3{,}409$ and $3{,}410$;
panel organization is as in Figure~\ref{fig:hk97-d3d6-a}.}
\label{fig:hk97-d3d6-d}
\end{figure}

\clearpage
Increasing $D$ retains somewhat more image-specific variation in the PCA
stage.  The forced projection, however, is driven by sparsely occupied extreme
bins: in different directions the selected bin contained as few as one oracle
observation.  Since the expected threshold count remains below one, the forced
outputs cannot be used to select an optimal value of $D$.  At the copied
ellipse calibration, strict use with $q_{\min}=30$ would require approximately
$q_{\min}/(b\Gamma_{\delta_{\mathrm{loc}}})\simeq209{,}700$ independent oracle
images.  A practical experiment with the present stack therefore requires an
HK97-specific support calibration rather than a direct transfer of the ellipse
parameters.

\section{Finite-sample support estimation}\label{app:support-oracle}

For the empirical analysis, retain the common localization radius and bin
width $\delta$ used by Algorithm~3. For $u\in S^{D-1}$, write
\[
 q_u(t):=\Gamma(H_{t,u}),\qquad h_u:=h_K(u),\qquad
 m_\delta:=c_{\MM_0}\omega_d\delta^d.
\]
The following conditions ensure both a separated population crossing and
that the scan starting beyond the unit ball reaches that crossing:
\begin{equation}\label{eq:oracle-parameter-regime}
 \begin{gathered}
 0<\delta\le\min\{1/4,\sigma/2\},\qquad
 L_\delta:=\log(1/m_\delta)\ge1,\\
 r_\delta=\frac{\sigma^2L_\delta}{\delta}\ge2+4\delta.
 \end{gathered}
\end{equation}
For every fixed $\sigma>0$ and supplied lower-mass constants, these
conditions hold for all sufficiently small $\delta>0$.

\begin{lemma}[A separated population crossing]\label{lem:population-crossing}
Under~\eqref{eq:oracle-parameter-regime}, the threshold $\Gamma_\delta$
defined in Section~\ref{sec:Algo} satisfies the following bounds for every
$u\in S^{D-1}$:
\begin{align}
 q_u(h_u+r_\delta-2\delta)&\ge2\Gamma_\delta,
 \label{eq:population-crossing-high}\\
 q_u(h_u+r_\delta+2\delta)&\le\tfrac12\Gamma_\delta.
 \label{eq:population-crossing-low}
\end{align}
\end{lemma}

\begin{proof}
For $s\ge0$, formula~\eqref{formula for Gamma} gives
\begin{equation}\label{eq:directional-density-upper}
 q_u(h_u+s)\le\kappa_0^{-1}
 \exp\!\left(-\frac{s^2}{2\sigma^2}\right).
\end{equation}
Choose $x_u\in\MM$ with $\langle x_u,u\rangle=h_u$. The $\delta$-ball
about $x_u$ has projected mass at least $m_\delta$, and every point $x$
in that ball satisfies $\langle x,u\rangle\ge h_u-\delta$. Hence
\begin{equation}\label{eq:directional-density-lower}
 q_u(h_u+s)\ge\kappa_1^{-1}
 \exp\!\left(-\frac{(s+\delta)^2}{2\sigma^2}\right).
\end{equation}
We have $r_\delta\ge4\delta$ because $\delta\le\sigma/2$ and
$L_\delta\ge1$. Substituting $s=r_\delta-2\delta$ into
\eqref{eq:directional-density-lower} and dividing by $\Gamma_\delta$
gives a lower bound
\[
 \exp\!\left(L_\delta-\frac{\delta^2}{2\sigma^2}\right)
 \ge e^{7/8}>2.
\]
Similarly, substituting $s=r_\delta+2\delta$ into
\eqref{eq:directional-density-upper} gives an upper bound for the ratio
\[
 \exp\!\left(-L_\delta-\frac{2\delta^2}{\sigma^2}\right)
 \le e^{-1}<\frac12.
\]
\end{proof}

Let
\[
 j_0:=\lceil\delta^{-1}\rceil+1,\qquad
 j_{\max}:=\left\lceil\frac{1+r_\delta+2\delta}{\delta}\right\rceil,
 \qquad \mathcal J_\delta:=\{j_0,\ldots,j_{\max}\}.
\]
For every $u$ and $j\in\mathcal J_\delta$, define
\[
 p_{u,j}:=\mathbb P\{j\delta<\langle Y,u\rangle\le(j+1)\delta\},
 \qquad
 \widehat p_{u,j}:=F_{j\delta}-F_{(j+1)\delta},
\]
where the $F$'s in $\widehat p_{u,j}$ are the empirical frequencies in
Algorithm~3. Thus $\Gamma^{est}_j=\widehat p_{u,j}/\delta$.
Let $\widehat j_u$ be the first scanned index with
$\widehat p_{u,j}\le\delta\Gamma_\delta$, and set
\[
 \widehat h(u):=\widehat j_u\delta-r_\delta.
\]
The output of Algorithm~3 at a halfspace
$H=\{x:\langle x,u\rangle\le t\}$ is therefore
$t-\widehat h(u)$, an estimate of the signed separation $t-h_K(u)$.

\begin{theorem}[Uniform finite-sample support oracle]
\label{thm:uniform-support-oracle}
Assume~\eqref{eq:oracle-parameter-regime}, let $\eta\in(0,1)$, and let
$\mathcal V_1\subset S^{D-1}$ be a fixed finite nonempty set, chosen
independently of the oracle batch. If that batch contains
$N_{\mathrm{oracle}}$ independent observations and
\begin{equation}\label{eq:uniform-oracle-sample-size}
 N_{\mathrm{oracle}}\ge
 \frac{C}{\delta\Gamma_\delta}
 \log\!\left(\frac{C|\mathcal V_1|(r_\delta+3)}{\delta\eta}\right),
\end{equation}
then, with probability at least $1-\eta$, the scan in Algorithm~3 stops
at an index at most $j_{\max}$ for every $u\in\mathcal V_1$, and
\begin{equation}\label{eq:uniform-support-conclusion}
 \max_{u\in\mathcal V_1}|\widehat h(u)-h_K(u)|\le3\delta.
\end{equation}
Consequently, the error in its signed-separation output is at most
$3\delta$, simultaneously for these directions and every $t\in\R$.
Equivalently, the sufficient sample size in
\eqref{eq:uniform-oracle-sample-size} is
\begin{equation}\label{eq:uniform-oracle-expanded}
 \frac{C\sqrt{2\pi}\sigma}{m_\delta\delta}
 \exp\!\left\{
 \frac{\sigma^2}{2\delta^2}
 \left[\log\!\left(\frac1{m_\delta}\right)\right]^2
 \right\}
 \log\!\left(\frac{C|\mathcal V_1|(r_\delta+3)}{\delta\eta}\right).
\end{equation}
If $\mathcal V_1$ is a $\rho$-net with $0<\rho\le1$ and
$|\mathcal V_1|\le(3/\rho)^D$, the final logarithmic factor is at most
\[
 C\left[D\log(3/\rho)+
 \log\!\left(\frac{C(r_\delta+3)}{\delta\eta}\right)\right].
\]
\end{theorem}

\begin{proof}
Put $\tau_\delta:=\delta\Gamma_\delta$. Since $-1\le h_u\le1$ and
$j_0\delta>1$, every scanned bin lies beyond the support in direction
$u$, where $q_u$ is decreasing. Whenever
$j\in\mathcal J_\delta$ satisfies $j\delta\le h_u+r_\delta-3\delta$,
its upper endpoint is at most $h_u+r_\delta-2\delta$, so
Lemma~\ref{lem:population-crossing} gives
\begin{equation}\label{eq:witness-bin-high}
 p_{u,j}=\int_{j\delta}^{(j+1)\delta}q_u(t)\,dt
 \ge2\tau_\delta.
\end{equation}
On the other hand, define
\[
 j_u^+:=\left\lceil\frac{h_u+r_\delta+2\delta}{\delta}\right\rceil.
\]
The condition $r_\delta\ge2+4\delta$ gives
$h_u+r_\delta+2\delta\ge1+6\delta>j_0\delta$; also
$j_u^+\le j_{\max}$. Thus $j_u^+\in\mathcal J_\delta$, and
\begin{equation}\label{eq:far-bins-low}
 p_{u,j_u^+}\le\tfrac12\tau_\delta.
\end{equation}

For a bin of probability at least $2\tau_\delta$, the multiplicative
Chernoff bound gives
\[
 \mathbb P\{\widehat p_{u,j}\le\tau_\delta\}
 \le \exp(-cN_{\mathrm{oracle}}\tau_\delta).
\]
For a bin of probability at most $\tau_\delta/2$, the same bound, with
another absolute constant $c>0$, holds for
$\mathbb P\{\widehat p_{u,j}>\tau_\delta\}$.
Moreover,
\[
 |\mathcal J_\delta|
 =j_{\max}-j_0+1
 \le\frac{r_\delta+3}{\delta}.
\]
A union bound over the high-probability bins in
\eqref{eq:witness-bin-high} and the single bin $j_u^+$ for each
$u\in\mathcal V_1$ therefore has failure probability at most
\[
 \frac{C|\mathcal V_1|(r_\delta+3)}{\delta}
 \exp(-cN_{\mathrm{oracle}}\delta\Gamma_\delta),
\]
which is at most $\eta$ under~\eqref{eq:uniform-oracle-sample-size}.
No independence between different bins or directions is needed.

On this event the scan cannot stop at a bin whose lower endpoint is at
most $h_u+r_\delta-3\delta$, and it must stop by $j_u^+$. Consequently,
\[
 h_u+r_\delta-3\delta
 <\widehat j_u\delta
 \le j_u^+\delta
 <h_u+r_\delta+3\delta.
\]
Subtracting $r_\delta$ proves~\eqref{eq:uniform-support-conclusion}.
Finally,
\[
 \Gamma_\delta^{-1}
 =\frac{\sqrt{2\pi}\sigma}{m_\delta}
 \exp\!\left(\frac{\sigma^2L_\delta^2}{2\delta^2}\right)
\]
gives~\eqref{eq:uniform-oracle-expanded}.
\end{proof}

\begin{remark}
For any nonempty finite oracle batch, Algorithm~3 eventually reaches an
empty bin beyond the largest projected observation, so its scan always
terminates. The theorem supplies a deterministic upper bound
$j_{\max}$ on its stopping index on the stated high-probability event.
The factor $(\delta\Gamma_\delta)^{-1}$ in the sample bound reflects
bin probabilities of order $\delta\Gamma_\delta$, rather than the
hyperplane density alone.
\end{remark}

\begin{remark}[Practical variant used in the appendices]
The experiments use a variant of the support oracle with localization
radius $\delta_{\mathrm{loc}}$ and a separate bin width $0<b\le\delta_{\mathrm{loc}}$.
For this variant only, put $\delta=\delta_{\mathrm{loc}}$ in the tail
calibration and define
\begin{equation}\label{eq:oracle-grid}
 \mathcal J_{\delta,b}
 :=\left\{j\in\mathbb Z:
 r_\delta-1-2\delta-2b\le jb\le r_\delta+1+2\delta+b
 \right\}.
\end{equation}
After checking $|P|b\Gamma_\delta\ge q_{\min}>1$, this variant selects
the largest index in $\mathcal J_{\delta,b}$ whose empirical bin mass
is at least $b\Gamma_\delta$, and returns the support estimate
$(j+\tfrac12)b-r_\delta$. It reports insufficient sample if the check
fails, or failure if no bin qualifies. References to Algorithm~3 in the
experimental appendices mean this practical variant. Theorem~\ref{thm:uniform-support-oracle}
and the main guarantee concern the first-crossing routine stated in
Section~\ref{sec:Algo}; they do not assert a guarantee for the experimental
parameter choices.
\end{remark}

\begin{remark}[Conditioning on the PCA block]
The PCA subspace and its projected support are random. Conditional on the
PCA block, they are fixed and the independent oracle batch remains
independent and identically distributed. A directional net chosen in PCA
coordinates before observing that batch is therefore admissible in
Theorem~\ref{thm:uniform-support-oracle}. Orthogonal projection preserves
the lower-mass bound, so the theorem's failure probability is uniform in
the PCA block and remains valid after averaging over it.
\end{remark}

\section{Projection Algorithm}\lab{sec:5}

Let $\MM$ be as in Section~\ref{sec:HypDist}, set $K=\conv(\MM)\subseteq
\mathbb B^D$, and let $P_K$ denote Euclidean projection onto $K$.
We approximate this projection using estimated support values on a finite
directional net. Proposition~\ref{prop:support-to-projection} gives a
deterministic error bound in terms of the target norm, the net radius, and
the support-estimation error. Under the tuning and sampling conditions of
Theorem~\ref{thm:main}, this yields the stated high-probability guarantee
for a fixed fresh observation.

\subsection{A deterministic reduction from support estimation to projection}

For $v\in\mathbb B^D$, recall that
$h_K(v)=\sup_{x\in K}\langle x,v\rangle$ and set
$g_y(v)=\langle y,v\rangle-h_K(v)$.  The minimax theorem gives
\begin{align}
 \dist(y,K)
 &=\min_{x\in K}\max_{v\in\mathbb B^D}\langle y-x,v\rangle \\
 &=\max_{v\in\mathbb B^D}\min_{x\in K}\langle y-x,v\rangle
 =\max_{v\in\mathbb B^D}g_y(v).                 \label{eq:support-duality-proof}
\end{align}
The vector $v=0$ shows that the last maximum is nonnegative. If $y\notin K$
and $y^*=P_K(y)$, its maximizer is
$v^*=(y-y^*)/\|y-y^*\|_2$ and
$y^*=y-\dist(y,K)v^*$. If $y\in K$, the maximum is zero and the projection is
$y$ itself.

The following proposition separates the deterministic optimization question
from the statistical construction of the support oracle.

\begin{proposition}[Projection from an approximate support oracle]
\label{prop:support-to-projection}
Suppose $K\subseteq\mathbb B^D$ is nonempty, compact, and convex. Let
$\mathcal V=\{0\}\cup\mathcal V_1$, where $\mathcal V_1$ is a finite $\rho$-net of
$S^{D-1}$. Assume that $\widehat h(0)=0$ and
\[
 \max_{v\in\mathcal V_1}|\widehat h(v)-h_K(v)|\le a.
\]
For $y\in\mathbb R^D$, define
\[
 \widehat v\in\argmax_{v\in\mathcal V}
 \{\langle y,v\rangle-\widehat h(v)\},\qquad
 \widehat d=\max\{0,\langle y,\widehat v\rangle-
 \widehat h(\widehat v)\},
\]
and $\widehat y=y-\widehat d\,\widehat v$. With
$L_y:=\|y\|_2+1$ and $q:=L_y\rho+2a$, we have
\begin{equation}\label{eq:deterministic-projection-error}
 \|\widehat y-P_K(y)\|_2
 \le L_y\rho+a+\sqrt{2\,\dist(y,K)\,q}
 \le L_y\rho+a+\sqrt{2L_yq}.
\end{equation}
If $y\in K$, the sharper bound
$\|\widehat y-y\|_2\le a$ holds.
\end{proposition}

\begin{proof}
Because $K\subseteq\mathbb B^D$, the functions $h_K$ and $g_y$ are
$1$-Lipschitz and $L_y$-Lipschitz, respectively, on $\mathbb B^D$. Let
$d=\dist(y,K)$. Equation~\eqref{eq:support-duality-proof}, the net property,
and the oracle bound imply
\[
 d-L_y\rho-a\le\widehat d\le d+a,
 \qquad
 d-g_y(\widehat v)\le q.
\]
In particular, $|\widehat d-d|\le L_y\rho+a$.

Suppose first that $d>0$. Let
$v^*=(y-P_K(y))/d$. Since
$h_K(v)\ge\langle P_K(y),v\rangle$, for every unit vector $v$,
\[
 g_y(v)\le d\langle v^*,v\rangle.
\]
If $\widehat v=0$, then $\widehat y=y$ and the claimed bound follows from
$d=d-g_y(0)\le q$. Otherwise $\widehat v$ is a unit vector and, consequently,
\[
 \frac d2\|v^*-\widehat v\|_2^2
 =d(1-\langle v^*,\widehat v\rangle)
 \le d-g_y(\widehat v)\le q.
\]
The triangle inequality now gives
\[
 \|\widehat y-P_K(y)\|_2
 \le|\widehat d-d|+d\|\widehat v-v^*\|_2
 \le L_y\rho+a+\sqrt{2dq}.
\]
Also $d\le\|y-x\|_2\le\|y\|_2+1=L_y$ for any $x\in K$.

If $d=0$, then $g_y(v)\le0$ for every $v$ and $g_y(0)=0$. Hence the
maximized empirical value belongs to $[0,a]$, so
$\widehat d\le a$ and $\|\widehat y-y\|_2\le a$.
\end{proof}

Thus a deterministic target error $\epsilon$ is obtained, for example, on
the event $\|y\|_2+1\le L$ by choosing $\rho$ and $a$ so that
\begin{equation}\label{eq:support-oracle-accuracy-for-projection}
 L\rho+a+\sqrt{2L(L\rho+2a)}\le\epsilon.
\end{equation}
For Algorithm~2, maximizing over $\mathcal V_1$ and taking the positive
part of the maximum gives the same output as including $v=0$ in
$\mathcal V$. Theorem~\ref{thm:uniform-support-oracle} supplies the
required simultaneous support bound with $a=3\delta$. Thus the
deterministic result applies to the short pseudocode in Section~\ref{sec:Algo},
including when the target lies inside $K$.

\section*{Code and data availability}

The code and numerical outputs supporting the end-to-end ellipsoid experiment
are publicly available at
\url{https://github.com/harius80/manifold-denoising-experiment}.
The repository contains the fixed-seed Python implementation, pinned
dependencies, raw results from all $20$ repetitions, the summarized results,
the generated figure.  No external data set is
required for that simulated experiment.  The exploratory experiment in
Section~\ref{sec:hk97-exploratory} uses a local HK97 cryo-EM image stack; the
raw stack is not redistributed with the ellipsoid reproducibility repository.

\section*{Declaration on AI usage}

OpenAI’s ChatGPT and Codex were used as assistive tools during the preparation
and revision of this manuscript, including for editing prose and organization, drafting
\LaTeX{} and Python code, literature discovery, and suggestions for checking
mathematical arguments, computations, and internal consistency. The final
proofs are due to the authors, with one significant input from AI: it suggested using Dudley's entropy integral on the manifold directly rather than its convex hull thus both sharpening and simplifying the final result. The authors made all
scientific decisions and take full responsibility for the content and conclusions of the
manuscript.

\section*{Acknowledgments}
We are deeply grateful to Professors Peter Doerschuk and John Johnson for
making the HK97 cryo-EM image stack available and permitting us to report
results obtained from it.

\section*{Funding}
This work was supported by the Department of Atomic Energy, Government of
India (project number RTI4014), by the Infosys--Chandrasekharan Virtual
Center for Random Geometry at the Tata Institute of Fundamental Research
(TIFR), and a Google India Research Award. C.~Fefferman and
J.~Marty gratefully acknowledge support from AFOSR grant
FA9550-23-1-0273. M.~Lassas was partially supported by European Research
Council Advanced Grant 101097198 and by the FAME flagship of the Research
Council of Finland (grant 359186). The views and opinions expressed are
those of the authors only and do not necessarily reflect those of the
funding agencies or the European Union.

\bibliographystyle{plain}
\bibliography{manifold_denoising-v31-refs}

@article{VeeslerEtAl2014HK97,
  author  = {Veesler, David and Khayat, Reza and Krishnamurthy, Srinath and Snijder, Joost and Huang, Rick K. and Heck, Albert J. R. and Anand, Ganesh S. and Johnson, John E.},
  title   = {Architecture of a {dsDNA} Viral Capsid in Complex with Its Maturation Protease},
  journal = {Structure},
  volume  = {22},
  number  = {2},
  pages   = {230--237},
  year    = {2014},
  doi     = {10.1016/j.str.2013.11.007}
}

@article{GongEtAl2016HK97,
  author  = {Gong, Yunye and Veesler, David and Doerschuk, Peter C. and Johnson, John E.},
  title   = {Effect of the Viral Protease on the Dynamics of Bacteriophage {HK97} Maturation Intermediates Characterized by Variance Analysis of Cryo-{EM} Particle Ensembles},
  journal = {Journal of Structural Biology},
  volume  = {193},
  number  = {3},
  pages   = {188--195},
  month   = mar,
  year    = {2016},
  doi     = {10.1016/j.jsb.2015.12.012},
  url     = {https://doi.org/10.1016/j.jsb.2015.12.012}
}

@article{XuEtAl2018HK97,
  author  = {Xu, Nan and Veesler, David and Doerschuk, Peter C. and Johnson, John E.},
  title   = {Allosteric Effects in Bacteriophage {HK97} Procapsids Revealed Directly from Covariance Analysis of Cryo-{EM} Data},
  journal = {Journal of Structural Biology},
  volume  = {202},
  number  = {2},
  pages   = {129--141},
  month   = may,
  year    = {2018},
  doi     = {10.1016/j.jsb.2017.12.013},
  url     = {https://doi.org/10.1016/j.jsb.2017.12.013}
}

@article{9a533cee-7bd2-34ce-a747-8254219de6bc,
 ISSN = {00905364},
 URL = {http://www.jstor.org/stable/43556496},
 author = {Sourav Chatterjee},
 journal = {The Annals of Statistics},
 number = {6},
 pages = {2340--2381},
 publisher = {Institute of Mathematical Statistics},
 title = {A NEW PERSPECTIVE ON LEAST SQUARES UNDER CONVEX CONSTRAINT},
 urldate = {2024-11-19},
 volume = {42},
 year = {2014}
}

@book{vershynin2018high,
  title={High-dimensional probability: An introduction with applications in data science},
  author={Vershynin, Roman},
  volume={47},
  year={2018},
  publisher={Cambridge University Press}
}

@article{FIMN,
author = {Fefferman, Charles and Ivanov, Sergei and Lassas, Matti and Narayanan, Hariharan},
title = {Fitting a manifold of large reach to noisy data},
journal = {Journal of Topology and Analysis},
volume = {17},
number = {02},
pages = {315-396},
year = {2025},
doi = {10.1142/S1793525323500012},
URL = {https://doi.org/10.1142/S1793525323500012},
eprint = {https://doi.org/10.1142/S1793525323500012}
}

@article{large_noise,
  title={Fitting a manifold to data in the presence of large noise},
  author={Fefferman, Charles and Ivanov, Sergei and Lassas, Matti and Narayanan, Hariharan},
  journal={arXiv preprint arXiv:2312.10598},
  year={2023},
  url={https://arxiv.org/abs/2312.10598}
}

@InProceedings{belloni,
  title = 	 {Escaping the Local Minima via Simulated Annealing: Optimization of Approximately Convex Functions},
  author = 	 {Belloni, Alexandre and Liang, Tengyuan and Narayanan, Hariharan and Rakhlin, Alexander},
  booktitle = 	 {Proceedings of The 28th Conference on Learning Theory},
  pages = 	 {240--265},
  year = 	 {2015},
  editor = 	 {Grünwald, Peter and Hazan, Elad and Kale, Satyen},
  volume = 	 {40},
  series = 	 {Proceedings of Machine Learning Research},
  address = 	 {Paris, France},
  month = 	 {03--06 Jul},
  publisher =    {PMLR},
  url = 	 {https://proceedings.mlr.press/v40/Belloni15.html}
}

@article{Fan,
 ISSN = {03195724},
 URL = {http://www.jstor.org/stable/3315465},
 author = {Jianqing Fan},
 journal = {The Canadian Journal of Statistics / La Revue Canadienne de Statistique},
 number = {2},
 pages = {155--169},
 publisher = {[Statistical Society of Canada, Wiley]},
 title = {Deconvolution with Supersmooth Distributions},
 urldate = {2026-04-20},
 volume = {20},
 year = {1992}
}

@inproceedings{Shamir,
  author    = {Ohad Shamir},
  title     = {On the Complexity of Bandit and Derivative-Free Stochastic Convex Optimization},
  booktitle = {Proceedings of the 26th Annual Conference on Learning Theory (COLT 2013)},
  editor    = {Shai Shalev-Shwartz and Ingo Steinwart},
  series    = {JMLR Workshop and Conference Proceedings},
  volume    = {30},
  pages     = {3--24},
  year      = {2013},
  publisher = {JMLR.org}
  }

@inproceedings{Jan,
  title={Information-theoretic lower bounds for convex optimization with erroneous oracles},
  author={Singer, Yaron and Vondr{\'a}k, Jan},
  booktitle={Advances in Neural Information Processing Systems (NIPS)},
  volume={28},
  pages={3204--3212},
  year={2015}
}

@article{GenovesePeronePacificoVerdinelliWasserman2012a,
  author  = {Christopher R. Genovese and Marco Perone-Pacifico and Isabella Verdinelli and Larry Wasserman},
  title   = {Minimax Manifold Estimation},
  journal = {Journal of Machine Learning Research},
  volume  = {13},
  pages   = {1263--1291},
  year    = {2012}
}

@article{GenovesePeronePacificoVerdinelliWasserman2012b,
  author  = {Christopher R. Genovese and Marco Perone-Pacifico and Isabella Verdinelli and Larry Wasserman},
  title   = {Manifold Estimation and Singular Deconvolution under {Hausdorff} Loss},
  journal = {The Annals of Statistics},
  volume  = {40},
  number  = {2},
  pages   = {941--963},
  year    = {2012},
  doi     = {10.1214/12-AOS994}
}

@article{AamariLevrard2019,
  author  = {Eddie Aamari and Cl{\'e}ment Levrard},
  title   = {Nonasymptotic Rates for Manifold, Tangent Space and Curvature Estimation},
  journal = {The Annals of Statistics},
  volume  = {47},
  number  = {1},
  pages   = {177--204},
  year    = {2019},
  doi     = {10.1214/18-AOS1685}
}

@article{YaoSuLiYau2023,
  author  = {Zhigang Yao and Jiaji Su and Bingjie Li and Shing-Tung Yau},
  title   = {Manifold Fitting},
  journal = {arXiv preprint arXiv:2304.07680},
  year    = {2023}
}

@article{YaoSuYau2024,
  author  = {Zhigang Yao and Jiaji Su and Shing-Tung Yau},
  title   = {Manifold Fitting with {C}ycle{GAN}},
  journal = {Proceedings of the National Academy of Sciences},
  volume  = {121},
  number  = {5},
  pages   = {e2311436121},
  year    = {2024},
  doi     = {10.1073/pnas.2311436121}
}

@article{YaoXia2025,
  author  = {Zhigang Yao and Yuqing Xia},
  title   = {Manifold Fitting under Unbounded Noise},
  journal = {Journal of Machine Learning Research},
  volume  = {26},
  number  = {45},
  pages   = {1--55},
  year    = {2025}
}

\end{document}